\documentclass[11pt,reqno]{amsart}

\usepackage{amsmath,amsfonts,amssymb,commath}
\usepackage{graphicx,harvard}
\usepackage{xcolor}

\usepackage{multibib}

\def\pihat{{\hat\pi}}
\def\Fhat{{\hat F}}
\def\dhat{{\hat d}}

\newcommand{\Expect}{\mathbf E}
\newcommand{\E}{\mathbf E}
\newcommand{\R}{\mathbb R}
\newcommand{\N}{\mathbf N}
\newcommand{\Prob}{\mathbf P}

\newcommand{\sgn}{\textbf{sgn}}

\newcommand{\alphahat}{\hat{\alpha}}
\newcommand{\Zhat}{\hat{Z}}

\newcommand{\yhat}{\hat{y}}
\newcommand{\Ahat}{\hat{A}}

\newcommand{\hubar}{\overline h}
\newcommand{\supp}{{\sf supp}}
\newcommand{\argmax}{\arg\max}

\newtheorem{theorem}{Theorem}
[section]

\newtheorem{proposition}[theorem]{Proposition}

\newtheorem{definition}[theorem]{Definition}
\newtheorem{lemma}[theorem]{Lemma}

\newtheorem{remark}[theorem]{Remark}

\begin{document}

\title[Machine Learning]{\bf Machine Learning for Strategic Inference}

\author{In-Koo Cho}
\address{Department of Economics, Emory University, Atlanta, GA 30322 USA}
\email{icho30@emory.edu}
\urladdr{https://sites.google.com/site/inkoocho}

\author{Jonathan Libgober}
\address{Department of Economics, University of Southern California, Los Angeles, CA 90089 USA}
\email{libgober@usc.edu}
\urladdr{http://www.jonlib.com/}

\date{\today}

\thanks{We thank Juan Carrillo, Jason Hartline, Navin Kartik, Roger Moon, Xiaosheng Mu, Guofu Tan, Ashesh Rambachan, Joel Sobel, Erik Strand, and especially Grigory Franguridi, for helpful conversations and comments, and seminar audiences at AMETS, Emory, the Econometric Society World Congress, the NSF/NBER/CEME Conference on Mathematical Economics, Rochester, UC-San Diego, and USC. This project started when the first author was visiting the University of Southern California.   We are grateful for hospitality and support from USC. 
  Financial support from the National Science Foundation is
  gratefully acknowledged. }

\begin{abstract}
We study interactions between strategic players and markets whose behavior is guided by an algorithm. Algorithms use data from prior interactions and a limited set of decision rules to prescribe actions. While as-if rational play need not emerge if the algorithm is constrained, it is possible to guide behavior across a rich set of possible environments using limited details. Provided a condition known as \emph{weak learnability} holds, Adaptive Boosting algorithms can be specified to induce behavior that is (approximately) as-if rational. Our analysis provides a statistical perspective on the study of endogenous model misspecification.
\end{abstract}

\maketitle

\section{Introduction}

The importance of algorithms in guiding economic behavior is already significant, and likely to only be more so in the years to come. But since a number of economic phenomena rely crucially on the presence of rational individuals on both sides of a particular interaction, an open question is whether traditional rational models apply to such situations. Of course, economists recognize that people often fail to act rationally, with certain consistent failures having empirical implications. But the extent to which algorithms are susceptible to errors is a separate issue, and one that should be addressed for economists to speak to the increasing number of applications where algorithm design plays a central role. 

This paper introduces a framework to study the question of whether and when algorithms can approximate rational behavior. In our model, a rational, strategic player (who we refer to as a \emph{sender}) chooses a strategy when interacting with an algorithm that prescribes actions to a stream of short lived actors (who we refer to as \emph{receivers}).\footnote{The terminology of sender and receiver is to highlight the role of our model's timing; but many of our applications are beyond where these labels are traditionally applied.} A distinguishing feature of our exercise is our focus on the problem of \emph{strategic inference}: Specifically, we assume that the sender commits to a strategy which maps states into actions, and so a rational receiver would update beliefs about the best reply after observing the sender's action. A rational receiver would thus make an inference regarding a payoff relevant state using knowledge of the sender's strategy. On the other hand, the algorithm has access to observations on what transpired in previous interactions. We are interested in comparing the rational receiver's strategy with the strategy induced by the algorithm, with a focus on determining when rationality can be replicated. We are particularly interested in the case where the algorithm seeks to provide these recommendations without using details of the sender's objective or the particular setting at hand (for instance, as a sales platform might when designing an algorithm to be used for a variety of different products). In other words, we are interested in finding an algorithm capable of inducing rationality under as wide a set of environments as possible. 

In our model, the algorithm produces recommendations using data from relevant interactions in the past (where data consists of sender actions and ex-post payoffs). These recommendations are determined by finding a \emph{best fitting} decision rule. By best fitting, we mean that there is minimal error, with errors being weighted according to some specific objective. The main assumption here is that the algorithm can determine the best fitting rule from a particular \emph{set of decision rules}, which we refer to as a \emph{hypothesis class} (following the machine learning literature). A crucial limitation is that this hypothesis class is restricted and must be specified in advance, so that not \emph{every} feasible mapping from messages into actions can be fit to the data. Thus, there is no a priori guarantee that finding the best fitting rule within the given set yields the rational reply; whether this property holds will depend upon the sender's strategy. 

Our theoretical question can be phrased as follows: Do these limitations of algorithms inhibit the ability to prescribe actions which are (approximately) rational? We show that, while constraints may be exploited by strategic actors, an algorithm designer with particular capabilities can induce the as-if rational outcome in equilibrium. The answer to this question thus depends on what we assume the algorithm is capable of. Our contribution is to identify what some of these capabilities are. 

Constraints on algorithms of the kind in our model are often studied in the machine learning literature, which typically treats the data generating process as exogenous. Our goal, however, is to perform a similar algorithm design exercise, but in a strategic setting. To make sense of the restrictions on classifiers that can be fit to data, it may be instructive to note that in typical machine learning problems, a simple prediction (for instance, a ``yes-no'' recommendation) is sought for an observation among a very large set of possibilities. Seeking to find the correct recommendation for each one may be intractable or undesireable (given data limitations), and so a simpler set may be used as a baseline. On the other hand, it may still be possible to construct a new decision rule if the algorithm specifies how this should be done in advance. In our model, this takes the form of assuming the algorithm is limited in what can be fit to the data, but is otherwise flexible, in a way we will make precise below. 

One interpretation of this limitation is that the algorithm suffers from a form of model misspecification: the true optimal decision rule for a receiver may fall outside of the class of decision rules that can be prescribed by the algorithm. There are two notable differences from a standard model misspecification exercise, however. The first difference is that the algorithm in our framework is concerned explicitly with prescribing behavior, and not with the problem of inference per se. In the (currently very active) literature on model misspecification (see, for instance, \citeasnoun{EspondaandPouzo14}), a decisionmaker is assumed to be potentially incorrect regarding the set of possible parameters, but otherwise uses an optimally chosen decision rule. We, on the other hand, are not (directly) interested in learning the underlying parameters, but rather making an optimal \emph{prediction}. The second difference is that the extent to which the optimal prediction falls outside of the realm of considered models is endogenous in our setting. Since we allow algorithms to specify decision rules arbitrarily---instead constraining how models can be fit to the data---they are, in principle, able to expand the potential decision rules the receiver could use if it is specified how this should be done. As a result, the extent to which the algorithm is misspecified is endogenous to the constraints of the algorithm design problem. 

What should one expect to happen given these limitations of an algorithm? On the one hand, in order for the algorithm to be able to give non-degenerate predictions without using detailed knowledge of the particular parameters of the receiver's problem, a sufficiently rich set of classifiers should be used. We focus on cases where this criterion suggests using at least the set of \emph{single-threshold classifiers}, which conditions a recommendation only on which side of a threshold the observable messages lie. However, \emph{since our setting requires strategic inference on the part of receivers}, this class of hypotheses is susceptible to manipulation by a rational sender. For our purposes, \citeasnoun{Rubinstein93} identified the key economic force, studying a buyer-seller game where the buyer is restricted in the set of decision rules that can be utilized. Specifically, this paper showed that if a \emph{rational} decisionmaker is \emph{restricted} to use a single threshold classifier---i.e., one that makes the same decision on a given side of a fixed threshold---then the seller can price discriminate via a particular form of randomization which ``fools'' these buyers into making a decision which is suboptimal given the realized price.\footnote{The reasoning behind this result is as follows. First, the optimally chosen classifier chosen can do strictly better than simply randomizing the guess, implying that the seller can exploit the incentives of the buyer in order to manipulate the decision rule. On the other hand, it is impossible for threshold rules to implement the optimal decision with probability 1 when this rational rule is non-monotone in the price. The first point implies the buyer trades off against errors, and the second point implies that the tradeoff falls short of the fully rational response. As a result, the seller can force a different decision than would be rationally optimal for these buyers (with arbitrarily high probability).} Our framework nests \citeasnoun{Rubinstein93} as a special case, but considers more general environments as well. 

Our analysis elucidates a tension between the ability to fit \emph{rich} and \emph{coarse} sets of models. As \citeasnoun{Rubinstein93} shows, if a decisionmaker is limited in the decision rules that can be utilized, then there is a potential for exploitation. In order to combat this temptation, one may seek to add more possible replies to be fit to the data; in other words, to make the hypothesis class richer. Indeed, a decisionmaker could prevent the particular instance of exploitation he highlights by doing so. However, fitting richer decision rules may have other undesirable consequences, and may still fail to prevent a slightly more elaborate strategy from succeeding at exploitation. Above, we mentioned that this view is common in the machine learning literature; finding the best fitting model within a set of models may be computationally demanding if this set is very large. A goal of our paper is to highlight this tradeoff between fitting coarse models---which have attractive statistical properties, but poor behavioral properties---and rich models, for which the sitution is reversed. 

Our proposed solution is to use the Adaptive Boosting algorithm (\citeasnoun{SchapireandFreund12}), which specifies exactly how to construct a decision rule as a weighted combination of classifiers, with the weights specified by the algorithm. The algorithm requires (repeatedly) fitting a classifier to some distribution over prices and outcomes, from some set of baseline classifiers. 

Returning to the particular question at hand, the requirement on the set of classifiers able to be fit is called \emph{weak learnability}, and it is significantly less demanding than requiring all possible rational replies to be specified. We seek to highlight that this requirement is necessary and sufficient to overcome the problem of model misspecification mentioned above, i.e., the gap between the set of decision rules that can be fit to the data and those that a rational receiver can utilize.  We provide results which show how to check it in several straightforward applications, particularly when resorting to single-threshold classifiers (which typically have natural interpretations). 

To summarize, the answer to our theoretical question is that rationality can be ensured with the ability to (a) find a best-fitting decision rule from a class which satisfies weak learnability, and (b) combine such classifiers in a particular way (specified in advance). It is worth emphasizing one technical difference---due to our focus on a strategic inference problem---between our exercise and similar ones considered in computer science or machine learning where these issues have received more attention. In principle, the rational decision in our model is \emph{not} observed if the sender uses a strategy that does not reveal the state given an observation. In a lemons problem, for instance, it may be that ``low quality'' is observed at some price, but that ``high quality'' is in fact more likely and that correspondingly a rational buyer (receiver) would choose a ``buying'' action.  Therefore, in our problem, the payoff-maximizing decision must be \emph{inferred} and constructed by the algorithm. One of the main results of this paper is that this added difficulty does not change the qualitative desirable properties of the algorithm, which we show using results from large deviations theory---though this does induce some modifications on precisely how good of an approximation the algorithm is able to guarantee. 

Returning to the discussion of \citeasnoun{Rubinstein93}, we see that the issue with single threshold classifiers is that they are not \emph{strong learners} (i.e., they cannot ensure the optimal decision is taken with probability 1 following any price), even though they are \emph{weak learners} (i.e., they can outperform random guesses when chosen optimally). The remarkable property of the Adaptive Boosting algorithm is that weak learnability is sufficient to construct a classifier that yields a similar guarantee as under a model class satisfying strong learnability. It is interesting that part of the intuition for the main result in \citeasnoun{Rubinstein93}---which relies upon the buyer being able to strictly improve payoffs beyond a trivial default to induce a particular decision rule---exactly tells us how to overcome the main conclusion, once we have the algorithm in hand. 

At first glance, it appears that there is a significant gap between decision rules satisfying weak learnability and those which induce rational replies. Rationality requires, in principle, very rich decision rules to be used, and for the performance of them to leave very little room for error. Weak learnability does not, and only requires a uniform improvement over a random guess. It is therefore perhaps surprising that in our exercise, the turns out to be no gap at all. Due to weak learnability, the apparent gap in rationality caused by the limitation in the decision rules that can be fit to data can be overcome by a clever choice of algorithm. The result is that the algorithm can induce rational behavior without knowing anything beyond the observed data from past interactions. In contrast, \emph{strong} learnability (i.e., prescribing the optimal action with high probability) will usually require precise knowledge of the sender's strategy.  

We briefly mention that the algorithm design problem we study accommodates a rich possible action space, even with the \emph{same} restrictions in the decision rules that can be fit to data. In particular, a version of the weak learnability condition in settings with two possible receiver actions also applies to settings with an arbitrary finite number of actions. This is in sharp contrast to many other papers in the large literature on ``decisionmakers as statisticians''  (reviewed below), which use similar motivation to study departures from rationality. These papers have typically focused on the binary action case. This limitation is very natural---many of the key results from machine learning which arise when there are two possible predictions do not extend easily (or even at all) to the case of multiple actions. However, we can handle this in our problem, suggesting our algorithm is of broader interest. We believe that this extension is important, as it shows our conclusions do not hinge on other artificial limitations on the environment. 

Our exercise provides formalism within which machine learning methods can be applied to answer new questions relevant to microeconomic theorists, and visa versa. Our model is deliberately abstract, in order to provide general principles guiding when the problem of model misspecification can be overcome. One key message is that while it is not possible to guarantee that rationality emerges for arbitrarily data generating process, it \emph{is} possible if the data generating process is endogenous (due to the strategic player) to the statistical algorithm. This argument requires some additional steps using incentives of the actors to demonstrate that the resulting output does in fact correspond to what is traditionally thought of as subgame perfection. This endogeneity issue makes the problem no longer a pure statistical exercise. The modifications our analysis requires extend beyond the initial need to show that it is possible to do better than random guessing in this environment. As our analysis elucidates, AdaBoost is capable of handling a \textit{particular} kind of unboundedness in the cardinality of the action space. It is thus necessary to discipline the environment further in order to achieve our results.

\section{Literature}

This paper takes the framework of PAC learnability, familiar from machine learning, and applies it to a strategic setting. Within economics, this agenda is most closely related to the literature on learning in games when behavior depends on a statistical method. The single-agent problem is a particular special case, and this case is the focus of  \citeasnoun{AlNajjar09} and  \citeasnoun{AlNajjarandPai2014}. However, since we are focused on a strategic setting,  the data the algorithm receives is \textit{endogenous} in our setting. In contrast, their benchmarks correspond to the case of exogenous data. This problem is also studied in  \citeasnoun{Spiegler2016}, who focuses on causality and defines a solution concept for behavior that arises from individuals fitting a directed acyclic graph to past observations. More recently, \citeasnoun{Zhaoetal2020} take a decision-theoretic approach in a single-agent setting with lotteries, showing how a relaxation of the independence axiom leads to a neural-network representation of preferences. 

Taking these approaches to games, the literature has still for the most part focused on settings where the interactions between players is \textit{static}, ruling out the main environments we are interested in here.\footnote{By itself the distinction may not immediately seem significant---after all, a Nash equilibrium in an extensive form game involves choosing a strategy to best respond to the opponent, and is usually stated as a single (and thus static) choice. However, the additional restriction to binary action or 0-1 prediction problems makes nesting our problem less straightforward.} In contrast, our setting is a simple, two-player (and two-move) sequential game. We also note that much (though not all) of this literature focuses on binary prediction problems, whereas we discuss how to specify algorithms in the general finite action cases as well. \citeasnoun{CherryandSalant2019} discuss a procedure whereby players' behavior arises from a statistical rule estimated by sampling past actions. This leads to an endogeneity issue similar to the one present in our environment, i.e., an interaction between the data generating process and the statistical method used to evaluate it.   \citeasnoun{EliazandSpiegler2018} study the problem of a statistician estimating a model in order to help an agent take an action, motivated as we are by issues involved with the interaction between rational plays and statistical algorithms.  \citeasnoun{Liang2018} also focuses on games of incomplete information, asking when a class of learning rules leads to rationalizable behavior. Studying model selection in econometrics,  \citeasnoun{OleaOrtolevaPaiPrat2019} consider an auction model and ask which statistical models achieve the highest confidence in results as a function of a particular dataset.\footnote{On the question of algorithms in particular, one concern is that the algorithm design problem may be susceptible to bias or induce unwanted discrimination when implemented, relative to rationality. See \citeasnoun{Rambachanetal2020} for an analysis of these issues and how they may be overcome.}

On the other hand, the literature on learning in extensive form games has typically assumed that agents experiment optimally, and hence embeds notion of rationality on the part of agents which we dispense with in this paper. Classic contributions include  \citeasnoun{FudenbergandKreps1995},  \citeasnoun{FudenbergandLevine1993} and  \citeasnoun{FudenbergandLevine06}. Most of this literature has focused on cases where there is no exogenous uncertainty regarding a player's type, and asking whether self-confirming behavior emerges as the outcome. An important exception is  \citeasnoun{FudenbergandHe2018}, who study the steady-state outcomes from experimentation in a signalling game. While a rational agent in our game would need to form an expectation over an exogenous random variable, signalling issues do not arise because our sender has commitment.  

Perhaps closest in motivation is the computer science literature studying how well algorithms perform in strategic situations, as well as how rational actors may respond when facing them.  \citeasnoun{Bravermanetal2018} consider optimal pricing of a seller repeatedly selling to a single buyer who repeatedly uses a no-regret learning algorithm. They show that, on the one hand, while a particular class of learning algorithms (i.e., those that are \emph{mean-based}) are susceptible to exploitation, others would lead to the seller's optimal strategy simply being to use the Myersonian optimum. \citeasnoun{Dengetal2019} also study strategies against no-regret learners in a broad class of games without uncertainty, and consider whether a strategic player can guarantee a higher payoff than what would be implied by first-mover advantage. \citeasnoun{Blumetal2008} consider the Price of Anarchy (i.e., the ratio between first-best welfare and worst-case equilibrium welfare), and show in a broad class of games that this quantity is the same whether players use Nash strategies or regret-minimizing ones. \citeasnoun{Nekipelovetal2015} assume players in a repeated auction use a no-regret learning algorithm, making similar behavioral assumptions as we do here. Their interest is in inferring the set of rationalizable actions from data.
 
While our motivation is very similar---and indeed, we seek to incorporate several aspects of this literature's conceptual framework---there are three notable differences. First, this literature typically assumes particular algorithms or principal objectives (such as no-regret learning) which differ from traditional Bayesian rationality. In contrast, we maintain a Bayesian rational objective for the seller, and also focus on an algorithm designer seeking to maximize the expected payoffs of agents. Second, we focus on relating the incentives of the rational player and \emph{the algorithm's capabilities}, and study the extent to which different assumptions on the algorithm design problem influence the task of approximating rationality. Our main result articulates how different action spaces for the algorithm designer yield different results regarding whether and when the outcome will approximate the rational benchmark. Lastly, our general framework focuses on settings with strategic inference---that is, where the payoffs following a given principal action are state-dependent---and thus covers a set of single-agent applications which extend beyond particular pricing settings, where most (though admittedly not all) of this literature has focused. In particular, the settings discussed in this literature do not cover Lemons markets settings (which \citeasnoun{Rubinstein93} falls under, for example) or Persuasion, which form our primary starting point.\footnote{An important exception is \citeasnoun{CamaraEtAl2020}, who study an environment covering many of our same applications such as Bayesian Persuasion. However, they still maintain the other two distinguishing features, focusing on a regret objective for the principal, as well as particular no-regret assumptions for the agent. Still, we emphasize that both our paper and theirs focuses on environments where the principal/sender chooses a state-dependent strategy. This leads to the aforementioned endogeneity between the data generating process (induced by the principal) and the choices of the algorithm/learner---this emerges due to the fact that the same sender action may induce two distinct replies from the algorithm following two distinct sender strategies. In their setting, this endogeneity motivates the use of ``policy-regret'' as an objective for the principal (due to their reinforcement learning approach to the principal's problem). While we do not use a regret objective for the principal, see \citeasnoun{AroraEtAl2012} and \citeasnoun{AroraEtAl2018} for more on the differences between these notions.} As a result, new technical issues (e.g., dealing with residual uncertainty in the correct actions) is not addressed in these papers to our knowledge.  Despite these differences, our hope is that this paper inspires further connection between the economics literature on decisionmakers as statisticians and the computer science literature on strategic choices against classes of algorithms. It appears to us that these results from computer science have not yet been fully appreciated in economics.

\section{(Sender-Receiver) Stage Games}

We first describe the stage game interaction in which the algorithm designer seeks to prescribe actions on behalf of myopic actors (who may be, for instance, receivers, buyers, or agents, depending on the particular setting of interest). The stage game features a strategic actor as well. That said, our exposition in this section addresses neither how the strategic actor chooses her strategy, nor how the algorithm is determined. This is done in Section \ref{subsection: machine game}, which describes the interaction which yields these objects and the relevant objective for each actor. 

\subsection{Actions and Parameters}

The stage game is a sender-receiver game in which an informed sender makes the first move.
We often call the sender \textit{the (informed) principal}, and the receiver \textit{the agent}, as our lead example is built on the informed principal problem of \citeasnoun{MaskinTirole92}. However, our model also describes a sender-receiver game with sender commitment, as in \citeasnoun{KamenicaGentzkow2011}.   

Let $\Theta$ be the set of types endowed with a prior distribution $\pi$ which is common knowledge among players. This type is payoff relevant to both the sender and the receiver. Define $\pi(\theta)$ as the probability that type $\theta\in\Theta$ is realized. Throughout the paper, we only consider $\pi$ with finite support. Conditioned on the realized value of $\theta\in\Theta$, the sender takes an action $p\in {\mathcal P}\subset {\mathbb R}^n$ where ${\mathcal P}$ is a compact subset of ${\mathbb R}^n$. Our analysis in most of the paper will assume further that $\abs{\mathcal{P}} < \infty$, although we discuss how to modify this assumption in Section \ref{sect: discretization} (to allow for, for instance, continuous distributions).  A strategy of the Sender is:
\[
\sigma : \Theta\rightarrow \Delta({\mathcal P}),
\]
where $\Delta(X)$ denotes the set of probability distributions over a set $X$. We let $\Sigma$ denote the set of feasible strategies for the sender, and importantly, assume that this strategy is determined (and committed to) in the Algorithm Game described in Section \ref{subsection: machine game}. Conditioned on $p$ (but not $\theta$), the agent chooses $a\in A$ according to
\[
r : {\mathcal P}\rightarrow \Delta(A).
\]
We assume $\abs{A} < \infty$, and in our analysis we treat the case of $\abs{A}=2$ and $\abs{A} > 2$ separately. The payoffs  of the sender and the receiver from $(\theta,p,a)$ are
$u(\theta,p,a)$ and $v(\theta,p,a)$.

The timing of the moves in the stage game is as follows:

\begin{itemize} 
\item[$S_1.$] An exogenous state $\theta\in\Theta$ is realized according to $\pi$, with only the sender observing the realized state $\theta$.

\item[$S_2.$] The sender's action $p\in {\mathcal P}$ is realized according to $\sigma(p : \theta)$.

\item[$S_3.$] The receiver takes action $a \in A$ conditioned on $p$.

\item[$S_4.$] Payoffs are realized according to $u(\theta,p,a)$ and $v(\theta,p,a)$.
\end{itemize}

For instance, if we interpret $p=(p_1,\ldots,p_n)$ as a contract, and $a\in A=\{-1, 1\}$ as ``reject'' ($a=-1$) or ``accept'' ($a=1$), the stage game is a model of the informed principal (\citeasnoun{MaskinTirole92}).   If $p$ is interpreted as a message sent by a worker, and $a\in A$ as the wage paid by the firm, then the stage game becomes a signaling game (\citeasnoun{Spence73}). For now, we place no further restrictions on $u(\theta,p,a)$ and $v(\theta,p,a)$, though these are often implicit in the economic problem of interest.

\subsection{Payoffs and the Rational Benchmark} Describing the outcomes of the above interactions when the receiver is rational is a familiar exercise. In this case, his optimization problem is
\[
\max_{a\in A}~~ \sum_{\theta \in \supp \pi} v(\theta,p,a)\pi(\theta : p)
\]
where $\pi(\theta : p)$ is the posterior probability assigned to $\theta$ conditioned on $p$.   If $p$ is used with a positive probability by $\sigma$, then $\pi(\theta : p)$ is computed by Bayes rule:
\[
  \pi(\theta:p) =\frac{\sigma(p:\theta)\pi(\theta)  }{
\sum_{\theta'}\sigma(p:\theta')\pi(\theta')    }.
\]
We define the \emph{rational label} to denote the receiver's strategy were they rational. More precisely, the optimal response is a function of the chosen $\sigma \in \Sigma$ and the realized $p \in \mathcal{P}$:
\[
y^R :\Sigma\times {\mathcal P}\rightarrow A
\]
is a solution to the following optimization problem:
\[
\sum_{\theta \in \supp \pi} v(\theta,p,y^R(\sigma,p))\pi(\theta: p) \ge \sum_{\theta \in \supp \pi} v(\theta,p,a)\pi(\theta: p) \qquad\forall a\in A
\]
where $\pi(\theta:p)$ is computed via Bayes rule whenever $\sum_\theta\sigma(p:\theta)\pi(\theta)>0$.\footnote{For a fixed $\sigma$, $y^R(\sigma,\cdot) : {\mathcal P}\rightarrow A$ is a strategy of the agent, satisfying sequential rationality.}

Define $\sigma^R$ as a best response of the sender against a Bayesian rational receiver with perfect foresight:
\[
 \sum_{\theta,p,a} u(\theta,p,a)\sigma^R(p:\theta) y^R(\sigma^R,p)\pi(\theta)
 \ge
 \sum_{\theta,p,a} u(\theta,p,a)\sigma(p:\theta) y^R(\sigma, p)\pi(\theta) \qquad\forall\sigma \in \Sigma.
\]
By the construction, $(\sigma^R,y^R)$ constitutes a perfect Bayesian equilibrium in the stage-game with a rational receiver.

\subsection{Examples of Stage Games} \label{RubinsteinExplanation}

Before proceeding to the description of the algorithm game, we describe a few of the stage game interactions that are of primary interest.  We will return to these later in order to illustrate the incentives for each player. 

\subsubsection{Insurance}
\label{Insurance}

The following is borrowed from \citeasnoun{MaskinTirole92}. Suppose that the principal (sender) is a shipping company seeking to purchase insurance from an insurance company, an agent (receiver) that is seeking to delegate the decision of whether to offer the terms put forth by the shipping company. The principal seeks insurance every period, but faces risk (e.g., due to the location of shipping demand) that is idiosyncratic every period. 

In this case, we imagine the principal choose terms within some compact set $\mathcal{P} \subset \mathbb{R}^{2}$, where $p=(x,q)$ denotes a policy which provides a payment $x$ in the event of a loss, and costs an amount $q$. If $\theta \in \{L,H\}$ (with $L < H$)  denotes the probability of a loss, then the principal's utility is: 

\begin{equation*} 
u(\theta, p, a) = \begin{cases} (1- \theta) f(I-q)+  \theta f(I-q-L+x) & a=1 \\ (1-\theta) f(I)+ \theta f(I-L) & a = -1 \end{cases},
\end{equation*} 

\noindent for some concave $f$. The agent's utility is: 

\begin{equation*} 
v(\theta,p, a) = \begin{cases} q- \theta x & a=1 \\ 0 & a=-1 \end{cases} 
\end{equation*}

It is natural to consider $\mathcal{P}$ whereby, against a rational buyer, the principal would seek a high level of insurance when risk is high (i.e., $\theta = H$), and avoid insurance when risk is low (i.e., $\theta=L$). In contrast, the agent's payoff may be decreasing in the quantity of insurance when $\theta=H$, while increasing in the quantity of insurance when $\theta=L$. 

\subsubsection{Labor Market Signaling}

Our framework is general and can be expanded to cover other settings as well. Let us consider a labor market signaling model.  Here, the ``receiver'' takes the role of the firm and the ``sender'' takes the role of the worker from the Spence signalling model (as in, for instance, \citeasnoun{MaskinTirole92}).  The true state is the productivity of the worker $\theta\in \Theta=\{ H, L\}$ where $\pi(H)=\pi(L)=\frac{1}{2}$: $H>L$.  Conditioned on $\theta$, a worker chooses $p$ which we interpret as education level.   Her strategy is
\[
\sigma : \Theta\rightarrow {\mathcal P}\subset {\mathbb R}_+.
\]
The payoff function of the sender is
\[
  u(\theta, p,a )=  a -\frac{p}{\theta+1}
\]
We abstract away the competition among multiple firms in the labor market.
Conditioned on $p$, the labor market wage is determined according to the expected productivity $\Expect(\theta : p)$ conditioned on $p$.  The firm has to pay the worker the equal amount of the expected productivity because of (un-modeled) competition among firms.  The receiver's goal is to make an accurate forecast about the expected productivity of the worker.  The payoff of the receiver is
\[
  v(\theta,p,a) = -(\theta-a)^2
\]
If the support of $\sigma(p:H)$ is disjoint from the support of $\sigma(p:L)$, $\sigma$ is a separating strategy.   If a separating strategy is an equilibrium strategy, then the equilibrium is called a separating equilibrium.   We often focus on the Riley outcome, which maximizes the ex ante expected payoff of the principal among all separating equilibria.

\subsubsection{Monopoly Market}
\label{Monopoly Market}

In the first two examples, only the sender has private information.
We can allow the stage game interaction to feature additional parameter observed only by the receiver, denoting this by $i$. We denote the sender's payoff by $u(\theta, p, a, i)$, and the receiver's payoff by $v(\theta, p, a, i)$.    For example, the sender may interact with some receivers who are algorithmic, and others who are fully rational.   The agent knows whether he is algorithmic or fully rational.  The principal does not observe the type of an agent, but only knows the probability distribution.  Indeed, the setting of \citeasnoun{Rubinstein93} features such a dichotomy, as we discuss. 

While \citeasnoun{Rubinstein93} differs expositionally, we review the key ideas and describe how it falls under our framework.   Suppose $\theta \in \Theta=\{L, H\}$.  $v_\theta$ is the marginal utility of the good where $v_H > v_L>0$.   The prior probability distribution is $\pi(H)=\pi(L)=\frac{1}{2}$.   The seller observes $\theta$.  The seller chooses a price $p\in  [v_L, v_H]\subset {\mathcal P}\subset {\mathbb R}$, conditioned on $\theta\in \Theta$.  The action of a buyer is $a\in A=\{ -1, 1\}$.   A buyer responds to $p\in {\mathcal P}$ by purchasing $(a=1)$ or not purchasing $(a=-1)$ the good at $p$.   

The seller is facing a unit mass of infinitesimal buyers, who can be either type 1 or type 2.   The proportion of type 1 buyer is $r\in (0,1)$.   A buyer observes his type, but the seller does not observe the type of a buyer.  The buyers differ in terms of the cost of sales.    If $\theta=L$, the product costs $c_L$ for the seller regardless of the types of the buyer.   If $\theta=H$, the product costs $c_i$ to serve type $i$ buyer $(i\in \{ 1,2\})$.   We assume
\begin{eqnarray}
  && c_1 > v_H > c_2 > v_L >c_L    \label{eq: lemon} \\
  && r c_1 + (1-r) c_2 > v_H      \label{eq: severe lemon}
\end{eqnarray}
so that the agent is exposed to the lemon's problem.   We focus on $\pi$ supported on parameters which satisfy \eqref{eq: lemon} and \eqref{eq: severe lemon}.

A buyer generates utility only if he purchases the good, whose payoff function is
\[
  v(\theta,p,a,i) =\begin{cases} 0    & \text{if } \ a=-1 \\
    v_\theta -p        & \text{if } \ a=1.
    \end{cases}
  \]
  regardless of $i$.
The payoff of the seller is
\[
  u(\theta,p,a,i) =\begin{cases}
    0    &  \text{if } \ a=-1   \\
    p- c_L   & \text{if } \ a=1, \ \theta=L   \\
    p- c_i & \text{if } \ a=1, \ \theta=H.
    \end{cases}
\]
The unique equilibrium strategy of the seller is
\[
  \sigma^R(\theta) =\begin{cases}
    v_L   &  \text{ if } \ \theta=L \\
    v_H   &  \text{ if } \ \theta=H.
    \end{cases}
\]
  The buyer's equilibrium strategy is
\[
    y^R(\sigma^R,p) =\begin{cases}
1  & \text{if } \ p \le v_L\\
-1  & \text{if } \ p >v_L.
      \end{cases}
\]
The trading occurs only if $\theta=L$, and therefore, the equilibrium is inefficient.   Note that the construction of $y^R$ requires a precise information about $v_L$.

\subsection{Introducing Time}

Our question of interest is whether the receiver can learn the rational label $y^R$, if the stage game is repeated over time.  As an intermediate step toward defining algorithm games, we describe our approach and assumptions involved with this step. In the next section, we discuss the algorithm choice that occurs on top of this. 

By \emph{expanded stage} game, we refer to a repetition of the stage game interaction, played over discrete time $t=1, 2, \ldots$, where the stage game interactions described previously occur at every $t \geq 1$.
Let $(\sigma_t,r_t)$ be the pair of strategies by the sender and the receiver in period $t$.\footnote{For now, we are intentionally vague about the strategy space of each player in the expanded stage game, as this is described in the next section.} The true state $\theta$ is drawn IID across periods according to $\pi$ and the pair $(p_t,a_t)$ of the price and the action in period $t$ is selected by
$(\sigma_t,r_t)$.   In this case, the expected payoff of the sender is\footnote{Our results are most elegantly stated in the undiscounted limit. Prior versions of this paper considered the case where future payoffs were discounted at rate $\delta < 1$; the main lessons remain valid for $\delta$ sufficiently large, although there are some added technical difficulties in the analysis of Section \ref{sect: discretization} this introduces.}
\begin{equation} 
  \lim_{T\rightarrow\infty}\Expect\frac{1}{T}\sum_{t=1}^T u(\theta_t,p_t,a_t)\pi(\theta_t)\sigma_t(p_t:\theta_t)r_t(a_t:p_t)
  \label{eq: seller long run}
\end{equation}
and the expected payoff of the receiver is
\begin{equation}
\lim_{T\rightarrow\infty}\Expect\frac{1}{T}\sum_{t=1}^T v(\theta_t,p_t,a_t)\pi(\theta_t)\sigma_t(p_t:\theta_t) r_t(a_t:p_t). 
\label{eq: buyer long run}
\end{equation} 

\section{Algorithm game}  \label{subsection: machine game}
Having outlined the basic timing of moves, we now describe the ``super'' game which determines the player's strategy in the expanded stage game. We refer to this as an \emph{algorithm} game.     Throughout this paper, we assume that the sender (principal) is fully rational, but the strategic choice of the receiver (agent) must be delegated to an algorithm.

\subsection{Choices of Algorithms}

We will refer to the strategy a receiver uses---which is output by the algorithm at every time---as a \emph{classifier}, in line with the machine learning and computer science literature: 

\begin{definition} A classifier is a function
  \[
    \gamma : P \rightarrow A .
  \]
  This may additionally be referred to as either a strategy or a forecasting rule.
\end{definition}

In order to construct the classifier, the algorithm faces some computational constraints. More precisely, we assume that there is a fixed set of classifiers $\mathcal{H}$ (referred to as the \emph{hypothesis class}) for which the algorithm can solve the following problem: 

\begin{equation} 
\min_{h \in \mathcal{H}}~~~ \sum_{p} \mathbf{1}[h(p)=y(p)]L(p), \label{eq:lossformulation}
\end{equation}

\noindent for an arbitrary function $L$ and function $y : \mathcal{P} \rightarrow A$. We refer to this step as finding the \emph{best fitting hypothesis}. We can think of $L$ as being the cost of misclassifying a particular observation, which may vary. Note that, since we can add arbitrary constants to $L$ and normalize so that it sums to 1 over all $p$, it is equivalent to assume the algorithm can solve

\begin{equation} 
\max_{h \in \mathcal{H}}~~~ \sum_{p} \mathbf{1}[h(p)=y(p)]d(p), \label{eq:probabilityformulation}
\end{equation}

\noindent for a probability distribution $d$ over $p$. This provides an alternative interpretation, regarding the classifier seeking to make the correct guess with the highest possible probability. 

We treat the process of finding the best fitting hypothesis as a black box. The purpose of this paper, however, is to understand how the algorithm designer might utilize from additional capabilities, and across a variety of environments. One question is which kinds of additional capabilities are necessary. The main ones we will discuss are: 

\begin{itemize} 
\item Constructing labels based on observations, 

\item Creating classifiers derived from solutions to the above maximization,

\item Changing observations of $p_{t}$ to $\hat{p}_{t}$, if the data is generated by a randomized rule.

\end{itemize}

One hypothesis class is of particular interest. Let ${\sf H}(\lambda,\omega)$ be a hyperplane in ${\mathbb R}^n$: $\exists\lambda\in {\mathbb R}^n$ and $\omega\in {\mathbb R}$ such that
\[
{\sf H}(\lambda, \omega)=\left\{ p\in {\mathbb R}^n \ : \ \lambda p =\omega \right\}.
\]
Define ${\sf H}_+(\lambda,\omega)$ as the closed half space above ${\sf H}(\lambda, \omega)$:
\[
{\sf H}_+(\lambda, \omega)=\left\{ p\in {\mathbb R}^n \ : \ \lambda p \ge\omega \right\}.
\]

\begin{definition}
A single threshold (linear) classifier is a mapping
\[
h : {\mathcal P}\rightarrow A
\]
where $\exists a_+, a_-\in A$, $\lambda \in \R^{n}$ and $\omega \in \R$ such that
\[
  h(p)=\begin{cases}
    a_+  & \text{if } \ p\in {\sf H}_+(\lambda, \omega) \\
    a_-  & \text{if } \ p\not\in {\sf H}_+(\lambda, \omega).
    \end{cases}
  \]
\end{definition}

\begin{definition}
Let $\Gamma$ be the set of all classifiers, and $\tilde{\Gamma} \subset \Gamma$ denote a subset of classifiers.  A \textit{statistical procedure} or \textit{algorithm} is an onto function
\[
  \tau : \mathcal{D} \rightarrow \tilde{\Gamma},
\]
where $\mathcal{D}$ is a set of histories, $\mathcal{T}$ is the set of feasible algorithms (i.e., a subset of the set of functions from $\mathcal{D}$ into $\tilde{\Gamma}$).
\end{definition}

\noindent
What $\mathcal {D}$ consists of is very much problem specific.   In a typical learning model, we assume that the receiver observes the realized outcome $(p_t,a_t)$ in period $t$ but also can access some information about the performance of his choice $a_t$ to achieve his goal.   For example, if the goal of the receiver is to learn the rational label $y^R$, a natural candidate would be a sufficient statistics of the ex-post payoff $v(\theta_t,p_t,a_t)$.

One specification of $\mathcal{T}$ emerges from not having any restrictions on $\tilde{\Gamma}$ at all. In general, the set $\tilde{\Gamma}$ will be implicit in the description of the algorithm. Our main interest is in understanding which kinds of ${\mathcal T}$ enable the receiver to approximate the rational label $y^R$.

\subsection{Timing and Objectives}

An algorithm game takes the interaction in the stage game as a starting point, and considers the outcome when, instead of having the receiver's strategy emerging from Bayesian rationality, it instead emerges from fitting a model to past observations.

An algorithm game is a simultaneous move game under asymmetric information between the (rational) sender and the boundedly rational (``algorithmic'') receiver, built on the ``expanded'' stage game.     
\medskip

\begin{itemize}
\item[$A_{-1}.$] According to some prior distribution, nature selects the distribution $\pi$ of the underlying game from $\Pi$, where $\Pi$ is a subset of probability distributions over $\Theta$ with finite support.

\item[$A_0.$] Conditioned on realized $\pi$, the sender commits to some strategy $\sigma$.   The receiver (or alternatively, an entity acting on the receiver's behalf) commits an algorithm $\tau \in \mathcal{T}$ without observing the realized $\pi \in \Pi$.  

\item[$A_1.$]  The expanded stage game is played, with the receiver's strategy in each period $t$ being $\tau(D_{t})(p)$ (i.e., the action specified by the algorithm following sender's action $p$ at time $t$), with the algorithm adding the observation (which includes $p$ and ex-post utility following each receiver action) to the dataset at the end of each period.  
\end{itemize}

\medskip

\noindent These actions determine the realized payoffs by each player, as described in the previous section. Notice that we do not necessarily assume that any pair $\pi_{1}, \pi_{2} \in \Pi$ have intersecting or even overlapping support (though this is also certainly allowed). Correspondingly, we emphasize we do not assume $\Theta$ is itself finite, even though all $\pi \in \Pi$ have finite support. Additionally, since we only assume the algorithm observes the receiver's ex-post utility, $\theta$ itself need not ever observed by the algorithm. For instance, $\theta$ may reflect a production cost which only influences the sender's payoff. In this case, while the algorithm would observe the receiver's payoff from each action, they would not observe the seller's cost. 

The sender chooses $\sigma$ once and for all, with the action $p_t$ drawn i.i.d. over time.    On the other hand, the receiver's strategy in period $t$ is determined by the algorithm $\tau$ and history $D_t$ in period $t$.
The expression for the payoffs of the sender and the receiver are (\ref{eq: seller long run}) and (\ref{eq: buyer long run}), respectively, where $r_t(\cdot:p)$ is given by $\tau(D_{t})(p)$.

We consider the objectives of the rational player and the algorithmic player separately. The former is straightforward; given a sequence $(\theta_{t}, p_{t}, a_{t})$, the rational player's payoff (i.e., the sender's payoff) function is simply the long run average expected payoff:
\[
  {\mathcal U}_s(\sigma,\tau)=  \lim_{T \rightarrow \infty} \frac{1}{T}
\sum_{t=1}^T  \Expect u(\theta_{t}, p_{t}, a_{t}),
\]
where $(p_t,a_t)$ is generated by $(\sigma,\tau)$ in period $t$ and the expectation is otherwise conditioned only on $\pi \in \Pi$ (recalling that $\theta$ is taken to be drawn IID).   The objective of the rational sender is to maximize ${\mathcal U}_s$ by choosing $\sigma$, conditioned on $\pi \in \Pi$.
The payoff function of the algorithm player is also the long-run average expected payoff:
\[
  {\mathcal U}_r(\sigma,\tau)=  \lim_{T \rightarrow \infty} \frac{1}{T}
\sum_{t=1}^T  \Expect v(\theta_{t}, p_{t}, a_{t}).
\]
Note that implicitly we the players do not discount future payoffs, we call the algorithm game an algorithm 

We are interested in comparing the outcomes induced by the algorithm and the rational label $y^{R}(\sigma,p)$ introduced in the last section. We note that the comparison is potentially unfair because algorithms are more constrained in the decision rules that can be used. We therefore introduce a notion of rationality reflecting these limits:  

\begin{definition} \label{def:CR}
An algorithm $\tau$ is \textbf{constrained rational}, if $\forall\epsilon,\delta>0$, $\forall\sigma$, $\exists T$ such that $\forall t\ge T$,
  \[
\Prob \left( \sum_{\theta \in \supp \pi}   v (\theta,p,\tau(D_t)(p))\sigma(p:\theta)\pi(\theta) \ge
\max_{h\in{\tilde\Gamma}} \sum_{\theta \in \supp \pi}  v (\theta,p,h(p))\sigma(p:\theta)\pi(\theta) -\epsilon\right) \ge 1-\delta,
  \]

\noindent with the probability referring to uncertainty over $D_{t}$. An algorithm $\tau$ is \textbf{fully rational} if $\tilde{\Gamma}$ is replaced by the set of all $h : \mathcal{P} \rightarrow A$.
\end{definition}

\noindent The ``constrained'' qualifier is due to the limits on the strategies that can be chosen by the receiver. A fully rational receiver would choose $y^{R}(\sigma,p)$ in the stage game; a constrained rational algorithm yields actions are as optimally as possible, given that its output must be within the expanded model class $\tilde{\Gamma}$.   We often regard $\gamma \in {\tilde\Gamma}$ as a forecasting rule and $\tau$ as a formal procedure to construct a (strong) forecasting rule.  

In later sections, we will also discuss an important performance criterion of an algorithm is PAC (Probably Approximately Correct) learnability (\citeasnoun{Shalev-ShwartzandBen-David14}).

\begin{definition} \label{def:PAC}
Algorithm  $\tau$ is PAC (Probably Approximately Correct) learnable if
$\forall\sigma\in\Sigma$,  $\forall\epsilon >0$, $\exists T$ such that $\forall t\ge T$
  \[
\Prob \left( \tau(D_t)(p)\ne y^R(\sigma,p) \right)<\epsilon,
  \]
with the probability referring to uncertainty over $D_{t}$ and the realized $p$.  
\end{definition}

\noindent A key difference between Definitions \ref{def:CR} and \ref{def:PAC} is that the latter is a condition on the actions themselves and the decision rule, yet the former is a condition on the utility. In order to learn the equililbrium outcome $y^R(\sigma^R,\cdot)$, $\sigma^R$ must be a best response to the decision rule induced by the algorithm in the long run.

\begin{definition}
An outcome $({\overline\sigma},\tau)$ of the algorithm game \textbf{emulates  $(\sigma^R,y^R)$ of the underlying stage game,} if ${\overline\sigma}=\sigma^R$ and
$\tau$ is fully rational.
\end{definition}

The substance of the definition is that $\sigma^R$ is a best response to $\tau$.  Then, along the equilibrium path of the algorithm game, the receiver behaves as if he perfectly foresees $\sigma^R$ and responds optimally subject to the feasibility constraint imposed by ${\tilde\Gamma}$.

\subsection{Specifying $\mathcal{T}$} 
Our main interest will be in the case where $\mathcal{T}$ is restricted to emerge as the outcome of an \emph{ensemble algorithm}. 

\begin{definition} \label{def:ensemble}
  Classifier $H$ is an ensemble of ${\mathcal H}$ if $\exists h_1,\ldots,h_K\in {\mathcal H}$ and $\alpha_1,\ldots,\alpha_K\ge 0$ such that
  \[
    H(\sigma, p)= \argmax_{\hspace{-7mm} a} \sum_{k=1}^{K} \alpha_k \mathbf{1}[a=h_{k}(\sigma, p)]
  \]
\end{definition}

\noindent Without loss of generality, we can assume that $\sum_{k=1}^K\alpha_k=1$, since if not we can simply divide by this sum and obtain the same classifier.  We can interpret $H$ as a weighted majority  vote of $h_1,\ldots,h_K$.  An ensemble algorithm constructs a classifier through a linear combination classifiers from $\mathcal{H}$.   Since the final classifier is constructed through a basic arithmetic operation, one can easily construct an elaborate classifier from rudimentary classifiers.   Ensemble algorithms have been remarkably successful in real world applications (\citeasnoun{Dietterich00}). 

The algorithms produce an output ensemble classifier according to a recursive scheme: 

\begin{itemize} 
\item First, the loss function in (\ref{eq:lossformulation}), say $L_{1}$, or probability distribution in (\ref{eq:probabilityformulation}), say $d_{1}$, is taken to treat all observed sender actions symmetrically---that is, $L_{1}(p)=d_{1}(p)=1/G$ where $G$ is the number of elements in the support of mixed strategy $\sigma$.

\item At each stage $k=1, \ldots$, the best fitting hypothesis is found by solving either (\ref{eq:lossformulation}) or (\ref{eq:probabilityformulation}). The best fitting hypothesis is referred to as $h_{k}$. 

\item The term $\alpha_{k}$ is then determined, possibly as a function of the objective of the best fitting hypothesis.

\item Depending on $h_{k}$ and $\alpha_{k}$, the loss function $L_{k}$ is updated to $L_{k+1}$ (or, in the case of distributions, $d_{k}$ is updated to $d_{k+1}$).

\item After repeating this iteration $K$ times, a classifier of the form of Definition \ref{def:ensemble} is output, which is used to determine the final choice of the receiver. 
\end{itemize}

The ability to use an ensemble algorithm allows additional richness in the set of classifiers that can be used. There remain, however, a number of challenges: 

\begin{itemize} 
\item Clearly, repeatedly solving the same problem will not yield different outcomes, and so to meaningfully expand $\mathcal{H}$ one needs to determine how to change the objective to be fit as well, and 

\item Weights must be specified in advance. 
\end{itemize}

\noindent Both of these are \emph{on top of} the need to potentially alter the observed $p_{t}$ and determining the labels $y_{t}(\sigma, p_{t})$ to use for the observations, since the observed utility-maximizing decision need not coincide with the rational one ex-post.

\begin{remark} 
The reader may still wonder why algorithm design is necessary in the first place. For instance, if $y^{R}(\sigma,p)$ is a single threshold rule, it may be surprising that simply fitting the optimal single threshold rule to the data is insufficient to emulate rationality. While it may be sufficient in some cases, it is not in general, and in particular the ability to emulate rationality does not follow from the rational reply being in $\mathcal{H}$; the reason is that it is necessary to be able to construct richer rules in order to deter the sender from deviating to exploit limitations in $\mathcal{H}$, which would prevent the receiver from choosing the rational reply ``off-path.'' This is articulated in Section \ref{sect: whyalg}, where we also clarify the role of taking a richer set of possible $\Pi$, to correspondingly justify choosing a sufficiently rich set of $\mathcal{H}$ to begin with. 
\end{remark}

\section{Main Results} 

We now present our main results, showing the existence of an equilibrium of the algorithm game where the rational reply is emulated. We begin with a preliminary observation, useful for understanding our subsequent analysis: PAC learnability is a sufficient condition for the algorithm game to have a Nash equilibrium emulating $(\sigma^R,y^R(\sigma^R,p))$.

\begin{proposition}  \label{pr: PAC learnable}
If $\tau$ is PAC learnable, then $(\sigma^R,\tau)$ is a Nash equilibrium of the algorithm game which emulates $(\sigma^R,y^R(\sigma^R,p))$.
\end{proposition}

\begin{proof}
If $\tau$ is PAC learnable, then the receiver learns $\sigma$ accurately in the long run.  Thus, the long run average expected payoff of the sender is
  \[
{\mathcal U}(\sigma,\tau)=\Expect_\theta u(\theta,\sigma,y^R(\sigma,\sigma(\theta)))
  \]
By the definition,  
\[
\sigma^R=\arg\max \Expect_\theta u(\theta,\sigma,y^R(\sigma,\sigma(\theta))).
\]
By PAC learnability,
\[
\lim_{t\rightarrow\infty} \Prob[ \tau(D_t)(p) \neq y^R(\sigma^R,p)] = 0,
\]
implying that $\E[v(\theta_{t}, p_{t}, a_{t})] \rightarrow \E[v(\theta_{t},p_{t}, y^{R}(\sigma^{R},p)]$ as $t \rightarrow \infty$. This implies the long run discounted payoffs are equal to those obtained against a rational player, and hence $(\sigma^R,\tau)$ constitutes a Nash equilibrium which emulates 
$(\sigma^R,y^R(\sigma^R,p))$.
\end{proof}

This observation suggests it suffices to show the PAC-condition holds; in that case, the sender would find it optimal to choose $\sigma^{R}$, and by definition it would not be possible for the receiver to outperform rationality. However, there are two main difficulties which we seek to emphasize: 

\begin{enumerate} 
\item First, it may be that $y^{R} \in \mathcal{H}$, and yet if $\mathcal{H}$ is limited then the rational outcome cannot be emulated without expanding the set of feasible decision rules, and 

\item Second, one still needs to specify how the algorithm should use the historical data in inferring the correct decision. 
\end{enumerate}

\noindent This section addresses each of these issues. We first consider the case where the receiver knows the values of $y^R(\sigma,p)$ $\forall (\sigma,p)$.  We then show, in Section \ref{sect: firstpart}, that the PAC-condition holds for an algorithm: 

\begin{proposition}  \label{prop: firstpart}  If the receiver knows the values of $y^R(\sigma,p)$ $\forall (\sigma,p)$, there exists an algorithm $\tau_A$ that is PAC learnable.   Thus, $(\sigma^R,\tau_A)$ is a Nash equilibrium of the algorithm game, which emulates $(\sigma^R,y^R)$.
\end{proposition}

We then turn to the case where the algorithm cannot observe $y^R(\sigma,p)$. This yields an algorithm $\tau_{\Ahat}$, which coincides with $\tau_{A}$ with the added step of inferring the labels. We show that we obtain an analogous result for this case in Section \ref{sect: secondpart}:

\begin{proposition} \label{cr: corollary} 
Suppose that $y^R(\sigma,p)$ is a strict best response $\forall\sigma$ but the receiver does not observe the values of $y^R(\sigma,p)$.   Then, there exists an algorithm $\tau_{\Ahat}$ that is PAC learnable.  $(\sigma^R,\tau_{\Ahat})$ is a Nash equilibrium of the algorithm game that emulates $(\sigma^R,y^R)$. 
\end{proposition}

In our analysis, the first step is to construct an algorithm that generates an accurate forecast in the long run.  The remaining step is to show whether the sender has an incentive to choose $\sigma^R$ against $\tau_{\Ahat}$ in the algorithm game, in Section \ref{sect: examplesreview}.

\subsection{Specifying the Algorithm and Weak Learnability (Proposition \ref{prop: firstpart})}  \label{sect: firstpart}
\subsubsection{Weak Learnability}

The sufficient condition which ensures we can approximate an arbitrary decision rule combining single-thresholds is \emph{weak learnability}. Roughly speaking, weak learnability says that the hypothesis class can outperform someone who had some very minimal knowledge of the truth of the hypothesis. That is, it must be that the hypothesis class can do better than a someone who made a random guess, which would be made correct with some arbitrarily small probability. While this may seem permissive---and indeed, it is certainly less stringent than requiring it can approximate the truth with high probability---the difficulty in achieving it is the fact that this guarantee must be uniform over all possible distributions. 

We formally define this as follows: 

\begin{definition}
Let $P(\sigma)$ be the support of $\sigma$. If $\hubar$ solves 
\[
\sum_{p\in P(\sigma)}d(p)\mathbf{1}[y(p)=\hubar(p)] \ge 
\sum_{p\in P(\sigma)}d(p)\mathbf{1}[y(p)=h(p)]   \qquad\forall h \in \mathcal{H},
\]
$\hubar$ is an optimal weak hypothesis.
\end{definition}

\begin{definition} \label{def: weaklearn}
If $\abs{A} =2$, a hypothesis class $\mathcal{H}$ is weakly learnable if, for every distribution $d$ over observations $p \in P(\sigma)$ and labels $y(p)$, the optimal weak hypothesis satisfies: 

\begin{equation*} 
\sum_{p\in P(\sigma)} D(p)(\mathbf{1}[y(\sigma,p)\neq h(p)]- \mathbf{1}[y(\sigma,p)=h(p) ]) \geq \rho.
\end{equation*}

If $\abs{A} > 2$, a hypothesis class $\mathcal{H}$ is weakly learnable if, for every distribution $d$ over observations $p \in P(\sigma)$ and labels $y(p)$, the optimal weak hypothesis satisfies: 

\begin{equation*} 
 \sum_{p \in P(\sigma)} \mathbf{1}[\hubar(p) \neq y(p)] d(p) \leq \sum_{p \in P(\sigma)} \E_{\tilde{y} \sim B}[(1-\rho)\mathbf{1}[\tilde{y} \neq y(p)]]d(p) ,
\end{equation*}
for some $\rho > 0$ and some distribution $B$ over $A$. 
\end{definition}

\noindent The second condition is a generalization of the first, though the first is perhaps more familiar from the machine learning literature (as most attention has focused on the two-label case). This condition reflects the idea that the classifier randomly guesses the label according to some distribution $B$, but is ``flipped to being correct'' with probability $\rho$. For the $\abs{A} > 2$ case, the right hand side describes the expected error in such a case, and the left hand side describes the error from the optimal weak classifier.  

If weak learnability fails, then \emph{no} recursive ensemble algorithm can be built to approximate $y(p)$ based on $\mathcal{H}$ alone.\footnote{For example, imagine $\mathcal{H}$ only consists of trivial classifiers. A corollary of a result in Appendix \ref{lemm: duality} is that these classifiers can do equally well as a random guesser. However, it is clear that they cannot do strictly better, as they are restricted to giving the same guess to all possible $p$, unlike a random guesser who is correct with an added probability $\rho$.} Perhaps more surprising is that it is tight, a fact which we discuss further in Section \ref{sect: conv}. For now, we simply mention that if we take $\mathcal{H}$, the set of single threshold classifiers is weakly learnable. 

\begin{proposition} \label{prop: hyperplaneweaklearn}   The set of single-threshold classifiers satisfies the weak learnability condition of Definition \ref{def: weaklearn}.
 \end{proposition}

\begin{proof}   See Appendix \ref{Proof of Proposition prop: hyperplaneweaklearn}.
\end{proof}

Our proof uses the important fact: Any hypothesis class that \emph{contains all label permutations} can at least match the random guess guarantee. The proof of this intermediate lemma uses a duality argument in order to show that no distribution can lead to a lower payoff when this condition is satisfied. Importantly, however, this is true for \emph{any} hypothesis class, including the trivial one. This observation allows us to show that the added richness of single-threshold classifiers is sufficient to provide the additional gain over random guessing. 

\subsubsection{From Weak Learnability to Decision Rules}  \label{sect: labelsprop}

For simplicity, we present the case where
\[
A=\{ -1, 1\}
\]
leaving the general case to the appendix.   The formal description of 
the algorithm takes two steps.    First, we describe an algorithm under the assumption that the receiver knows the value of $y^R(\sigma,p)$ $\forall p$. If $A$ contains two elements, 
the specification of the algorithm parameters coincides with the Adaptive Boosting algorithm $\tau_A$ of \citeasnoun{SchapireandFreund12}.  We first outline the parameters and then review, for completeness.  

The $k$th stage (initializing with the uniform distribution if $k=1$) starts with probability distribution $d_k(p)$ over the support of $\sigma$.   Define
\begin{equation}
\epsilon_k=\Prob_{d_k}\left( h_k(p) \neq y^R(\sigma,p) \right)
\label{eq: epsilon}
\end{equation}
as the probability that the optimal classifier $h_k$ at $k$ misclassifies $p$ under $d_k$.
If $\epsilon_k=0$, then we stop the training and output $h$ as the forecasting rule, which perfectly forecasts $y^R(\sigma,p)$.  

Suppose that $\epsilon_k>0$.   Define
\begin{equation}
\alpha_k =\frac{1}{2} \log \frac{1-\epsilon_k}{\epsilon_k}
\label{eq: alpha 2}
\end{equation}
The weak learnability of the single threshold rule implies that $\exists\rho>0$ such that
\[
\epsilon_k \le \frac{1}{2}-\rho \qquad\forall k\ge 1.
\]
Define for each $p$ in the support of $\sigma$, and each pair $(p,y^R(\sigma,p))$,
\[
d_{k+1}(p)=\frac{
d_k(p)\exp (-\alpha_k y^R(\sigma,p) h_k(p)) 
}{
Z_k
}
\]
where
\[
  Z_k=\sum_{p} d_k(p)\exp (-\alpha_k y^R(\sigma,p) h_k(p)).
\]
Given $d_{k+1}$, we can recursively define $h_{k+1}$ and $\epsilon_{k+1}$, both of which are functions of $d_{k+1}$ as per the above. \medskip

The decision of the receiver is based upon
\[
\tau_A(D_k)(p)=\arg\max_{a\in A}\sum_{t=1}^k\alpha_t {\mathbf 1}(h_t(p)=a)
\]
which is equivalent to
\[
\tau_A(D_k)(p)=\sgn\left[\sum_{t=1}^k\alpha_t h_t(p)\right]
\]
if $A=\{-1,1\}$, where $\sgn(x)$ is the sign of real number $x$.

Following \citeasnoun{SchapireandFreund12}, we can show that
\begin{equation}
\Prob\left( \tau_A(D_t)(p)=y^R(\sigma,p) \right) \ge 1-e^{-t\rho(G)}
\label{eq: SF convergence}
\end{equation}
for any mixed strategy $\sigma$, where $G$ is the number of elements in the support of $\sigma$.\footnote{A sketch of the proof is in Appendix \ref{Proof of Proposition pr: convergence}.}

\subsection{Inferring the Rational Label (Proposition \ref{cr: corollary})} \label{sect: secondpart}

Next, we drop the assumption that the receiver can observe $\sigma$ so that he can calculate the expected utility conditioned on $p$ in the support of $\sigma$:
\begin{equation}
  \sum_\theta v(\theta,p,a)\mu(\theta: p)
  \label{eq: ex post expected value}
\end{equation}
where $\mu$ is computed via Bayes rule, and therefore, knows the value of $y^R(\sigma,p)$.
If the receiver does now know $y^R(\sigma,p)$, then he cannot calculate $\epsilon_k$ in \eqref{eq: epsilon}.   We need to construct an estimator $\yhat_t(p)$ for $y^R(\sigma,p)$ from data $D_t$ available at the beginning of period $t$.
%%%%%
How we construct estimator $\yhat_t(p)$ depends upon the specific details of the rule of the game such as the available data and the variable of interest.   We require that $\yhat_t(p)$ satisfies a regularity property.

\begin{definition}
$\yhat_t(p)$ is a consistent estimator if $\yhat_t(p)$ converges to $y^R(\sigma,p)$ in probability as $t\rightarrow\infty$.
\end{definition}

We require that $\yhat_t(p)$ satisfies the large deviation property (LDP), which is a stronger property than consistency.

\begin{definition}
$\yhat_t(p)$ satisfies large deviation properties (LDP) if $\exists\lambda>0$ such that, $\forall p$ in the support of $\sigma$,
\begin{equation}
\limsup_{t\rightarrow\infty} -\frac{1}{t}\log \Prob\left( y^R(\sigma,p)\ne\yhat_t(p)  \right) \le \lambda. 
\label{eq: point LDP}
\end{equation}
\end{definition}

If an estimator satisfies LDP, the tail portion of the forecating error vanishes at the exponential rate, as the sample average of i.i.d. random variables converges to the population mean.   If an estimator fails to satisfy LDP, the finite sample property of the estimator tends to be extremely erratic (\citeasnoun{Meyn07}).  Most estimators in economics satisfy LDP.

In the three examples illustrated in Section \ref{RubinsteinExplanation}, the variable of interest is the probability distribution of the underlying valuation conditioned on $p\in {\mathcal P}$. Let $\pi(v: p)$ be the posterior distribution of $v$ conditioned on $p$.   If $v$ is drawn from a finite set, then $\pi(v:p)$ is a multinomial distribution.
Let $\pihat_t(v:p)$ be the sample average for $\pi(v:p)$.  We know that the rate function of $\pihat_t(v:p)$ is the relative entropy of $\pihat_t$ with respect to $\pi$ (\citeasnoun{DemboandZeitouni98})
\[
I_\pi=\sum_{v}\pihat_t(v:p)\log\frac{\pihat_t(v:p)}{\pi(v:p)},
\]
from which we derive $\lambda$ in \eqref{eq: point LDP}: $\forall\epsilon>0$, let $N_\epsilon(\pi)$ be the $\epsilon$ neighborhood of $\pi(v:p)$, and
\[
\lambda=\inf_{\pihat_t\not\in N_\epsilon(\pi)}I_\pi.
\]
Note that
\[
y^R(\sigma,p)\ne \yhat_t(p)
\]
only if $\pi$ and $\pihat_t$ prescribe differen actions.  Since $\pihat_t$ is a consistent estimator of $\pi$, the probability of two probability distributions prescribing two different actions vanishes.   The large deviation property of $\pihat_t$ implies that $\yhat_t(p)$ satisfies \eqref{eq: point LDP}, if $y^R(\sigma,p)$ is a strict best response.

By the concavity of the logarithmic function, $I_\pi$ is minimized if $\pi$ is a uniform distribution and
\[
\inf_{\pi}I_\pi >0.
\]
If $\abs{P}<\infty$ and $\abs{A}<\infty$, we obtain the uniform version of \eqref{eq: point LDP} with respect to the true probability distribution.   We state the result without proof for later reference.

\begin{lemma} Suppose that
$\yhat_t(p)$ is consistent and satisfies \eqref{eq: point LDP}.   Then,
$\exists\lambda>0$ such that
  \begin{equation}
\limsup_{t\rightarrow\infty} -\frac{1}{t}\log \Prob\left( y^R(\sigma,p)\ne\yhat_t(p) \ \ \forall p \ \text{in the support of} \ \sigma    \right) \le \lambda. 
\label{eq: LDP}
\end{equation}
\end{lemma}

We construct algorithm $\tau_{\Ahat}$ by replacing $y(\sigma,p)$ by $\yhat_t(p)$ in $\tau_A$ constructed in the previous section.    More precisely, let $f^y_t(p)$ be the empirical probability that $\yhat_t(p)=1$ at the beginning of period $t$. Thus, $\yhat_t(p)=-1$ with probability $1-f^y_t(p)$.   Given $\{d_t(p),\yhat_t(p)\}_p$, $h_t$ solves
\[
\max_{h\in {\mathcal H}}\sum_p h(p) d_t(p) [ 1\cdot f^y_t(p) - 1\cdot (1-f^y_t(p))] 
\]
and
\[
  {\hat\epsilon}_t=\sum_p d_t(p)
 \left[ f^y_t(p){\mathbf 1}(h(p)=1)+(1-f^y_t(p)){\mathbf 1}(h(p)=-1)\right].
\]
Using weak learnability, we can show that $\exists\rho>0$ such that
\[
{\hat\epsilon}_t\le \frac{1}{2}-\rho.
\]
Since $\yhat_t(p)$ has the full support over $\{ -1,1\}$ $\forall t\ge 1$,
\[
{\hat\epsilon_t}>0.
\]
Given an algorithm  $\tau_A$ with observed labels, we can therefore replace it with $\tau_{\Ahat}$ which involves inferring the labels $y^R(\sigma,\cdot)$, setting them equal to $\hat{y}_{t}(\cdot)$, for all $t \geq 1$.

With $\tau_{\Ahat}$, we can construct labels from data, and that for the hypothesis class of interest the weak learnability condition is satisfied. The last step to show the algorithm works, in the case where the set of possible $p$ has finite support, is that the output of the algorithm will indeed converge to the rational reply, as dictated by the labels, provided the weights are specified correctly.

\begin{proposition} \label{pr: convergence} 
Suppose that $\yhat_t$ satisfies uniform LDP and that $y^R(\sigma,p)$ is a strict best response $\forall p$.  Then, $\forall\sigma$ that randomizes over $G$ elements of ${\mathcal P}$, $\exists T$ and $\exists\rho (G)>0$ such that
\[
\Prob\left( \tau_{\Ahat}(D_t)(p)=y^R(\sigma,p) \ \forall t\ge T \right) \ge 1-e^{-t\rho(G)}.
\]
\end{proposition}

\begin{proof}
See Appendix \ref{Proof of Proposition pr: convergence}.
\end{proof}

The construction of $\yhat_t$ depends on the specifics of a problem, especially what data $D_t$ available in period $t$ contains.    In many interesting economic models, the algorithm for $\yhat_t$ needs the knowledge of $\Theta$ rather than simply the support of $\pi$, and $D_t$ can contain at least the {\em ordinal} information about the performance of the decision recommended by $\tau_{\Ahat}$.

Let us consider the insurance model illustrated in Section \ref{Insurance}.   The critical value is \eqref{eq: ex post expected value}.   Instead, suppose that the receiver can observe the average performance difference of two actions:
\begin{equation}
  \sgn\left( \sum_{k=1}^{t-1} v(\theta_k,p,1)-v(\theta_k,p,0) \right)
  \label{eq: ordinal information}
\end{equation}
in the past.\footnote{At the end of each period, the receiver is supposed observe the performance difference.    If not, we can devise an experimentation strategy to infer the average  performance difference following the idea of exploration and exploitation. }
That is,
\[
  \yhat_t(p)=\begin{cases}
    1   & \text{if } \ \sum_{t'=1}^{t-1} v(\theta,p,1)-v(\theta,p,0).\ge 0 \\
   -1   & \text{if } \ \sum_{t'=1}^{t-1} v(\theta,p,1)-v(\theta,p,0). < 0.
    \end{cases}
\]
Given a probability distribution over $y^R(\sigma,p)$, $\yhat_t(p)$ satisfies LDP:
$\exists \lambda>0$ such that
\[
\limsup_{t\rightarrow\infty}-\frac{1}{t}\log\Prob \left( \yhat_t(p)\ne y^R(\sigma,p)\right) \le \lambda.
\]
We know that the large deviation rate function over a binominal distribution is uniformly bounded from below (\citeasnoun{DemboandZeitouni98}).   Thus, we can choose $\lambda>0$ uniformly over all probability distribution over $y^R(\sigma,p)$.

The ordinal information \eqref{eq: ordinal information} about the average quality is necessary.   Without access to \eqref{eq: ordinal information}, the algorithm cannot estimate $y(\sigma,p)$, which is critical for emulating the rational behavior.   The information contained in \eqref{eq: ordinal information} is coarse, because the algorithm does not take any cardinal information about the parameters of the underlying game.    Without the cardinal information, the receiver cannot implement the equilibrium strategy of the baseline game, which is a single threshold rule.   Because the algorithm does not rely on parameter values of the underlying game, the algorithm is robust against specific details of the game, if the algorithm can function as intended by the decision maker.

\subsection{Discussion}

\subsubsection{Accommodating Multiple Actions}  \label{sect: conv}

We use the Adaptive Boosting algorithm, as introduced by \citeasnoun{SchapireandFreund12}, to specify the $\alpha_{k}$ weights and the updates if $\abs{A}=2$.    The original Adaptive Boosting algorithm only applies to the case of $\abs{A}=2$.    To handle the case of $\abs{A} > 2$, we appeal to a generalization introduced by \citeasnoun{MukherjeeSchapire2013}.

The $\abs{A}> 2$ algorithm works for the $\abs{A}=2$ case, with one minor drawback, which is that the learnability constant must be computed in advance. While our work shows an algorithm exists, the computation of the learnability constant is more indirect and hence explicitly finding a parameter that works is more difficult.   The arguments for these proofs follow from results in the machine learning literature (see \citeasnoun{SchapireandFreund12}), which we can apply to show that this algorithm can yield a response for which the misclassification probability vanishes.

The proof of Proposition \ref{pr: convergence} is stated for the general case.
The proof reveals that the rate at which the probability of misclassfication vanishes is determined entirely by the number of sender actions in the support of $\sigma$.   Thus, the algorithm is efficient in that it maintains an exponential rate of convergence (\citeasnoun{Shalev-ShwartzandBen-David14}).

\subsubsection{On the Necessity of Expanding $\mathcal{H}$}  \label{sect: whyalg} So far, our analysis has assumed that the set of initial classifiers $\mathcal{H}$ contains the set of single-threshold classifiers, we have shown that the rational reply can be guaranteed by an algorithm. We now show that this result requires the ability to construct algorithms to expand the set of possible decision rules, even if $y^{R}(\sigma, p) \in \mathcal{H}$, given sufficient richness in the set $\Pi$---recall that we do not necessarily assume $\abs{\Theta} < \infty$, so that different $\pi$ with non-overlapping support may be possible; that is, if the designer seeks to provide rational replies in a variety of different environments, one cannot simply find the best fitting hypothesis within $\mathcal{H}$ to emulate rationality.   

Indeed, it is straightforward to find conditions on $\Pi$, the set of possible distributions over $\Theta$ (all of which, we assume, have finite support), such that the optimal decision rule is of the threshold form. In fact, the algorithm designer may improve upon the rational reply given knowledge of $\pi$. To illustrate, suppose $u(\theta, p, a)$ is \emph{independent} of $\theta$, and weakly increasing (coordinatewise) in $p$, for each $a$ (the latter of which would hold if, for instance, $p$ were a menu of prices). The following simple result shows that in this case, at least (increasing) single threshold classifiers should be included: 

\begin{proposition} \label{prop: atleastsingle}
Suppose $u(\theta, p, 1) - u(\theta, p,0)$ is constant in $\theta$ and weakly concave in $p$. Suppose further that $u(\theta, p^{*}, 1) > u(\theta, p^{*},0)$. Then there exists a single threshold classifier which the algorithm could commit to using which ensures the strategic player chooses $p^{*}$ with probability 1.
\end{proposition}

\noindent The proof follows immediately from an observation that the set of $p$ at which the consumer chooses $a=1$ is convex under the conditions of the proposition.\footnote{See also \citeasnoun{GilboaSamet1989} for a similar observation on how the use of restricted decision rules can be advantageous.}

In order for the algorithm designer to improve upon a degenerate prescription to always choose $a=0$, Proposition \ref{prop: atleastsingle} suggests including at least single threshold classifiers which are increasing. Against the highlighted $\pi$, such prescriptions would give the receiver even higher commitment power than the rational benchmark. In order to maximize payoff against richer and richer $\Pi$, more and more classifiers should therefore be included to $\mathcal{H}$.

This raises the question of whether adding in these classifiers goes ``too far.'' Namely, in seeking to maximize payoff against a rich set of possible $\pi$, does this risk doing \emph{worse} against others?  In fact,  it may be that the receiver does \emph{worse} than the rational benchmark. 

\begin{proposition} \label{thm:fullsurplus}
Suppose that ${\tilde\Gamma}$ is the set of all single threshold classifiers, and suppose all $\pi \in \Pi$ has binary support. For any $\{\theta_{L}, \theta_{H}\}$ supporting $\pi$, suppose the following is satisfied: 

\begin{itemize} 
\item The sender's optimal $p-a$ pair when $\theta= \theta_{L}$ is $(p_{L}^{*}, 1)$ 

\item The sender's optimal $p-a$ pair when $\theta=\theta_{H}$ is $(p_{H}^{*}, 0)$, with $p_{L}^{*}  < p_{H}^{*}$. 

\item $v(\theta_{H}, p_{H}^{*},1) = v(\theta_{H}, p_{H}^{*}, 0)$,

\item $v(\theta_{L}, p_{L}^{*}, 1) \geq v(\theta_{L}, p_{L}^{*}, 0) $, and 

\item  $v(\theta, p, 1) -v(\theta, p, 0)$ increasing in $p$, for all $\theta$. 
\end{itemize}

Then a policy arbitrarily close to the sender's optimal $p-a$ pair is implementable, even if this differs from the rational outcome under $\sigma^{R}$. 
\end{proposition} 

A setting where this sender-optimal action strategy differs from the rational outcome was first studied, to the best of our knowledge, in \citeasnoun{Rubinstein93}. His setting satisfies the conditions of the proposition. Our proof adapts his arguments to the current setting (i.e., incorporating the statistical aspect of our exercise and beyond the application he considered, described above), and highlights the importance of counterveiling incentivse in driving the result. The reason the sender can profitably deviate in the previous proof is because the new $\sigma$ induces a \emph{non-monotone} response from the receiver optimally, even though this is not prescribed by $\sigma^{R}$. In contrast, decision rules with single-threshold classifiers must be monotone.\footnote{Even though $y^R(\sigma^R,p)$ is an element of ${\tilde\Gamma}$, $y^R(\sigma^*,p)$ is not an element of ${\tilde\Gamma}$.   Since $\sigma^*$ is the choice variable of the sender, the sender generates misspecification endogenously. }

In other words, we show that the sender can construct a strategy which ensures that the receiver's utility as a function of $p$ violates single-crossing. Now, if the sender were using the particular $\sigma^*$ from the previous proof, then the rational response can be achieved via a double-threshold classifier, since there are only three optimal sender choices. But on the other hand, if the receiver were restricted to using single- or double-threshold classifiers, then one could find another strategy whereby the optimal response would be to use a triple-threshold classifier, via a similar scheme. As long as the number of thresholds used is finite, a similar kind of exploitation would emerge.

\subsubsection{Accommodating Richer Principal Action Spaces} \label{sect: discretization}

While $\tau_{\Ahat}$ is designed to be robust against parametric details of the underlying problems, the algorithm is still vulnerable to strategic manipulation by the rational sender.   The proof of Proposition \ref{pr: convergence} reveals that the rate of convergence is decreasing as the number of sender actions in the support of $\sigma$ increases.   The sender can randomize over infinitely many messages to slow down the convergence rate arbitrarily.  That said, such manipulation would be short lived, and therefore have limited gains. Nevertheless, in order to ensure that there are only a finite number of observations that the algorithm may observe, it is necessary to augment the observation space so that the distribution facing the receiver can be treated as discrete. This section discusses how this modification can be done. 

\medskip

\noindent \emph{Approach One: Discretization} \hspace{5mm}  We describe how to revise $\tau_{\Ahat}$ accordingly to discretize the observation space.   Instead of processing individual actions, we let $\tau_{\Ahat}$ process a group of actions at a time, treating ``close'' actions as the same group.   In principle, we want to partition $\mathcal{P}$ into a set of half-open rectangles intervals with size $\lambda$.  More precisely, given some arbitrary $\lambda$, we can partition each dimension of a rectangle containing $\mathcal{P}$ into the collection of half open intervals of size $\lambda>0$ with a possible exception of the last interval:
\[
  P_0^{j}=[{\underline p}_, {\underline p}+\lambda),\ldots,
  P_{K_{j}^\lambda}^{j}=[{\underline p}+(K_{j}^\lambda-1)\lambda ), {\overline p}]
\]
where $K_{j}^\lambda$ is the number of elements in the partition and $j \in \{1, \ldots, n\}$ is a partiular dimension.

For each element in the partition, the algorithm receives an ordinal information about the average outcome from the decision, if it contains a sender action in the support of $\sigma$:
\[
  \yhat^\lambda_t(k) = a \text{  if  }  a = \argmax \sum_{p \in P_{k}} v(\theta, p, a)
\]
where $p$ in the support of $\sigma$ and $P_{k}$ is the product of partition elements.    Let $\tau^\lambda_{\Ahat}$ be the algorithm obtained by replacing $\yhat_t(p)$ in $\tau_{\Ahat}$ by $\yhat^\lambda_t(k)$.  Note that as $\lambda\rightarrow 0$, the size of the individual elements in the partition shrinks and $\tau^\lambda_{\Ahat}$ converges to $\tau_{\Ahat}$ for a fixed $\sigma$.  

Compared to $\tau_A$ and $\tau_{\Ahat}$, $\tau^\lambda_{\Ahat}$ takes only coarse information for two important reasons.   First, the algorithm cannot differentiate two $p$s which are very close.  This features makes the algorithm robust against strategic manipulation of the sender to slow down the speed of learning.   Second, the algorithm cannot detect the precise consequence of its decision, but only the ordinal information of the past decision, aggregated over time.    The second feature allows the algorithm to operate with very little information about the details of the parameters of the underlying game. \medskip

\noindent \emph{Approach Two: Smoothing}   \hspace{5mm}  Discretizing the action space as above is one way of ensuring that there are only a finite number of sender actions to worry about in the long run, and given a sufficiently fine discretization, any distinct $p$ is distinguished by the algorithm. However, in principle, close sender actions may still be quite far in terms of payoffs, and only be distinguished in the long run. That is, there is no guarnatee that for a fixed horizon, that the algorithm is not grouping too many $p$ possibilities. The issue is that the discretization approach uses no information about the receiver's payoff function. Our other alternative describes more explicitly how \emph{close} to rationality the receiver can achive, given some fixed discretization scheme.

The idea is the following: We add a small amount of noise to each observed $p$, with the amount of noise tending to 0 as the sample size grows large. Doing so allows us to show that the receiver perceives the sender's strategy to have the property that $\E_{\theta}[u(a, \theta, p(a) : p]$ is uniformly equicontinuous (as functions of $p$). As a result, if the receiver only seeks to use a strategy that is $\varepsilon-$optimal against $\sigma$, uniform equicontinuity implies that their best reply can essentially be collapsed within intervals. 

It will additionally be important that the algorithm does not seek to make predictions at $p$ values where the corresponding density would be estimated to be small. Hence a second step will be to determine whether a $p$ realization occur in a region with sufficiently large probability, where the ``sufficient'' amount will also tend to 0 as the amount of data grows large.

Formally, suppose the algorithm observes data $((p_{1},y), \ldots, (p_{n},y))$. Let $z_{\eta,i}$ be an independent random vector in the unit ball around 0 distributed according to the PDF:

\begin{equation*} 
\phi_{\eta}(z) = \frac{1}{K} \exp \left(- \frac{1}{1-\abs{z/\eta}^{2}} \right)\frac{1}{\eta^{\abs{A}-1}},
\end{equation*}

\noindent where $K$ is a constant which ensures $\phi_{\eta}$ integrates to 1. Our first augmentation is the following: 

\begin{itemize} 
\item  Replace the observed $p_{1}, \ldots, p_{n}$ with $\hat{p}_{1}, \ldots, \hat{p}_{n}$, where $\hat{p}_{i} = p_{i} + z_{\eta,i}$, with $z_{\eta, i}$ distributed according to the above. 
\end{itemize} 

\noindent Second, it turns out that the above smoothing operation only works if the density is sufficiently large. Otherwise, the smoothing noise has too much power.  

\begin{itemize} 
\item For any $\tilde{p}=(\tilde{p}_{a})_{a \in A \backslash a_{0}}$ drawn, estimate the event that $\tilde{\sigma}_{\eta}(\tilde{p}) < \gamma$ by fixing some $\delta$ small and determining whether menu(s) $p$ with $\max_{a \in A \backslash \{a_{0}\}} {\tilde{p}_{a}- p_{a}} < \delta$ occurs with frequency at least $(2\delta)^{\abs{A}-1} \gamma$. Recommend action $a_{0}$ for any such $p$.  
\end{itemize}

As $\delta \rightarrow 0$, the condition holds if the density is at least $\gamma$. Together with the previous, we can show that if the receiver instead observes noisy sender actions, the perceived sender's strategy is sufficiently well-behaved to maintain the appropriate convergence for the algorithm. 

\begin{proposition} 
Suppose the sender is restricted to choosing distributions which are either discrete or continuous with bounded density. Consider an algorithm which can ensure that an $\varepsilon$-rational label is PAC-learnable, for any arbitrary $\varepsilon > 0$ given a finite number of possible sender actions. Then there exists a smoothing operation which maintains PAC-learnability of $\varepsilon$-rationality, for every $\varepsilon > 0$. 
 \label{pr: Smoothing}
\end{proposition}

The idea of the proposition is to use the smoothing operation to show that the algorithm perceives that the sender uses a $\sigma$ such that $\E[u(a, \theta, p(a)) : p]$ is uniformly equicontinuous. Given that we seek $\varepsilon$-optimality, uniform equicontinuity allows us to essentially discretize the menu space, transforming the environment into a much simpler one. 

There are two important properties of the transformation which allows us to ensure this works. The first is that, defining $\tilde{\sigma}_{\eta}( \cdot : \theta)$ to be the perceived $p$ distribution of $p_{i} + z_{i}$, we have: 

\begin{equation*} 
D^{\alpha} \tilde{\sigma}_{\eta}(p : \theta) = \int_{P} D^{\alpha}\phi_{\eta}(p-\tilde{p}) \sigma(\tilde{p} : \theta) d\tilde{p},
\end{equation*}

\noindent so that $\tilde{\sigma}_{\eta}$ inherits the smoothness properties of $\phi_{\eta}$. The second is that, on any compact subset of $P$, we have $\sigma_{\eta}(\cdot : \theta) \rightarrow \sigma(\cdot : \theta)$ uniformly. Now, in order to obtain uniform continuity as $\eta \rightarrow 0$, it will be important that we can simultaneously ensure that the sender's strategy does not involve dramatic movements in the conditional probability. For instance, suppose the sender were to use the following strategy: 

\begin{equation*} 
\sigma(p : G) = p(\sin \left( \frac{1}{p} \right) +1),  \sigma(p : B) = p(\sin \left( \frac{1}{p} - \pi \right) +1),
\end{equation*}

\noindent defined on an interval $[0, \overline{p}]$ such that both densities integrate to 1. Then $\Prob[\theta=G : p]=1$ if $p= \frac{1}{(2k+1/2)\pi}$ for some $k \in \N$, and 0 if $p= \frac{1}{(2k + 1/2)\pi}$, for some $k \in \N$. As $k \rightarrow \infty$ (so that $p \rightarrow 0$), this oscillates infinitely often.

We handle the problem this example poses by only making non-degenerate predictions if the probability of using such sender actions is sufficiently high. That is, we ``ignore'' $p$ realizations which only occur with low probability according to an estimated density.\footnote{One may wonder why this trick works; for instance, we do not obtain the result when $\sigma(p : G) = \sin \left( \frac{1}{p} \right) +1,  \sigma(p : B) = \sin \left( \frac{1}{p} - \pi \right) +1.$ However, unlike the previous example, these will fail the continuity requirement on the sender's strategy space, which is needed in the proof. } Seeking to estimate the probability that all sender actions are within $\delta$ of $p$ in order to estimate the density is just one way of doing this step; for instance, one could estimate the CDF $\tilde{\sigma}_{\eta}(p)$, and use the estimated density to determine whether the observations should be thrown away. Ultimately, however, given the compact $\mathcal{P}$, we can minimize the probability that this is done by using sufficiently low thresholds. As a result, it has a vanishing impact on PAC-learnability, as well as the sender's expected profit.

\section{Review of Examples} \label{sect: examplesreview}

We now verify the implications of this observation on the sender behavior in our particular examples, showing that this results in the sender-preferred Stackleberg outcome is emerging. This requires us to verify the previously discussed conditions in the context of these applications.

\subsection{Informed Principal}

The decision problem of the agent is to identify each pair $(x,q)$ of payment $x$ and cost $q$ as an acceptable constract $(a=1)$ or not $(a=-1)$.   Without loss of generality, we can assume that the agent uses the single threshold linear classifier
induced by hyperplane
\[
{\sf H}(\lambda_x,\lambda_q,\omega) =\{ (x,q) : \lambda_x x +\lambda_q q =\omega \}
\]
and 
\[
  h(x,q)=\begin{cases}
    1    & \text{if } \ (x,q)\in {\sf H}^+(\lambda_x,\lambda_q,\omega) \\
    -1   & \text{otherwise.}
    \end{cases}
\]
We can construct $\tau_{\Ahat}$ by estimating $\Expect v(\theta,p,a)$ for each $(p,a)$.   

\begin{lemma}
  Suppose that $\sigma$ assigns a positive probability to $(x,q)$ where
  \[
\Expect (q-\theta x : (q,x))=0
\]
for $x>0$.  Then $\sigma$ is not a best response to $\tau_{\Ahat}$.
\end{lemma}

\begin{proof}
Let $(x,q)$ be some offer such that the agent is indifferent between accepting and rejecting, so that: 

\begin{equation*} 
q- \E[ \theta : (x,q)] x=0
\end{equation*}

The principal's expected payoff is found by taking the expectation of $u(\theta, (x,q),a)$ over all realizations of $x,q$. By the law of iterated expectations, this occurs if and only if the principal's payoff is maximized following \emph{each} realization of $(x,q)$. We claim the principal is \emph{not} indifferent between actions following any such $(x,q)$. Indeed, letting $\E[ \theta : (x,q)] =r$, indifference implies: 

\begin{equation*} 
 (1- r) f(I-rx)+  r f(I-L+(1-r)x)= (1-r) f(I)+ r f(I-L).
\end{equation*}

Note that equality holds if $f(y)=y$. This implies that both lotteries, whether or not the principal accepts, have the same expected values. However, if $f$ is concave, then since $I > I-rx > I-L +(1-r)x > I-L$, it must be that the left hand side is strictly greater than the right hand side. 

It follows that if indifference holds, the principal strictly prefers the agent accept the offer by slightly reducing $x$. 
\end{proof}

Following the same logic as in the previous example, we conclude that if $\sigma$ is a best response to $\tau_{\Ahat}$, then 
\[
\tau_{\Ahat}(D_t)(q,x)=y^R(\sigma, (x,q)) 
\]
with probability 1.  A best reply $\sigma$ to $\tau_{\Ahat}$ emulates $(\sigma^R,y^R(\sigma^R,p))$.

\subsection{Labor Market Signaling}

The firm's objective function is to forecast the productivity of the worker:
\[
v(\theta,p,a)=-(\theta-a)^2
\]
If $A$ is a real line, then
\[
y^R(\sigma,p)=\arg\max_{a\in A} \Expect_\theta \left[ v(\theta,p,a) : p,\sigma \right]
\]
where the posterior distribution over $\theta$ is calculated via Bayes rule from $\sigma$ and the prior over $\theta$.   Strict concavity of $v$ implies that $y^R(\sigma,p)$ is a strict best response $\forall\sigma,p$.

Without loss of generality, we consider a single threshold decision rule parameterized by $(a^+,a^-,p^0)$:
\[
  h(p)=\begin{cases}
a^+   & \text{if } \ p \ge p^0 \\
a^-   & \text{if } \ p < p^0. \\
    \end{cases}
  \]
Let ${\mathcal H}$ be the set of all single threshold decision rules.   In each round, $h_t$ solves
\[
\max_{h\in {\mathcal H}}\Expect_\theta \left[ v(\theta,p,a) : p,\sigma \right]
\]
if the data includes $\sigma$.  We construct $\tau_A$ accordingly.    If $\sigma$ is not observable by the algorithm, we estimate the posterior distribution of $\sigma$ conditioned on each $p$ to construct $\tau_{\Ahat}$.
If the agent learns $y^R(\sigma,p)$ eventually $\forall\sigma,p$, then the principal's choice $\sigma^R$ 
\[
  \Expect[ u(\theta,p,y^R(\sigma,p)) : \sigma, p]  
   =  \sum_\theta \sum_p u(\theta,p,y^R(\sigma,p))\sigma(p:\theta)\pi(\theta).
\]
If $\sigma^R$ entails separation by the high productivity worker, then the Riley outcome is the solution, that generates the largest ex ante expected surplus for the principal among all separating equilibria.   In order to satisfy the incentive constraint among different types of the principal, the principal with $\theta=H$ incurs the signaling cost.   If the signaling cost outweighs the benefit of separation, then $\sigma^R$ is the pooling equilibrium where both types of the workers takes the minimal signal.

The analysis is based upon the assumption that $y^R(\sigma,p)$ is a strict best response $\forall\sigma,p$.   As  $\abs{A}=J<\infty$, $y^R(\sigma,p)$ may not be a strict best response for some $\sigma$ and $p$.   Let us assume that
\[
A=\{a_0=0,a_1,\ldots,a_J\}
\]
and $a_i-a_{i-1}=\Delta>0$ and $a_F=1+\sup \Theta>0$.   Although $y^R(\sigma,p)$ may not be a strict best response for some $\sigma$ and $p$, the set of best responses contains at most 2 elements, which differ by $\Delta>0$.   Abusing notation, let $y^R(\sigma,p)$ be the set of best responses, if the agent has multiple best responses at $p$.
Applying the convergence result, we have $\exists T$ such that, $\forall t\ge T $,
\[
  \Prob\left(
  \exists y\in y^R(\sigma,p), y=\tau_{\Ahat}(D_t)(p) \right) < e^{-\rho t}. 
\]
For a sufficiently small $\Delta>0$, $\sigma^R$ is either a strategy close to the Riley outcome, or the pooling equilibrium where both types of the principals choose the smallest value of $p$.

\subsubsection{Monopoly Market}

In the model of \citeasnoun{Rubinstein93} illustrated in Section \ref{Monopoly Market},
suppose that type 1 buyer is an algorithmic player, while type 2 buyer is a rational player.    To simplify the model, we assume that type 2 buyer's decision is $y^R(\sigma,p)$ $\forall \sigma$.  Assume that type 1 buyer uses $\tau_{\Ahat}$.

$y^R(\sigma,p)$ is not a strict best response, if
\begin{equation}
  \Expect_\theta v(\theta,p,a,i)=0 \qquad\forall a\in A=\{1,-1\}, \forall i
  \label{eq: zero surplus condition}
\end{equation}
so that the agent is indifferent between accepting and rejecting $p$.
Thus, $\tau_{\Ahat}$ is not PAC learnable.

Still, the best response of the monopolistic seller against $\tau_{\Ahat}$ is $\sigma^R$.   The critical step is to show that a rational seller would not use any $\sigma$ which assigns positive probability $p>v_L$ satisfying \eqref{eq: zero surplus condition}.

\begin{lemma} \label{lm: strictly dominated}
Fix $\sigma$ which assigns $p >v_L$ with positive probability, satisfying
\begin{equation}
  \Expect_\theta v(\theta,p,1)\ge 0.
  \label{eq: positive surplus condition}
\end{equation}
Then, the ex ante expected profit of the principal against $\tau_{\Ahat}$ from $\sigma$ is strictly smaller than from $\sigma'$:
\[
{\mathcal U}(\sigma^R,\tau_{\Ahat}) > {\mathcal U}(\sigma,\tau_{\Ahat}).
\]
\end{lemma}

\begin{proof}
It suffices to show that 
if $p>v_L$ and $\Expect(v:p)-p\ge 0$, then the expected profit from $p$ is strictly less than $\pi_L v_L$.   We write the proof in \citeasnoun{Rubinstein93} for the later reference.
For any price $p$ satisfying
\[
\Prob(H :p) v_H + \Prob(L:p)v_L \ge p,
  \]
  the revenue cannot exceed
  \[
\Prob(H:p)  v_H+\Prob (L:p)  v_L 
  \]
  but the cost is
  \[
\Prob(H : p) (1-r) c_2 +\Prob(H:p) r c_1.
  \]
  Thus, the seller's expected profit is at most
  \[
\Prob(L:p)  v_L + \Prob(H:p) ((1-r)(v_H-c_2) +r(v_H-c_1))
  \]
Because of the lemon's problem,
\[
(1-r)(v_H-c_2) +r(v_H-c_1)<0
\]
and
\[
\Prob(H:p)>0
\]
to satisfy
\[
\Prob(H :p) v_H + \Prob(L:p)v_L \ge p > v_L.
  \]
Integrating over $p$, we conclude that the ex ante profit is strictly less than
$\pi_L v_L$.
\end{proof}

  Lemma \ref{lm: strictly dominated} implies that again $\tau_{\Ahat}$, the principal will not use $\sigma$ which assigns a positive probability to $p$ so that both 1 and -1 are best responses.   Thus, if $\sigma$ is a best response to $\tau_{\Ahat}$, then $y^R(\sigma,p)$ is a strict best response $\forall p >v_L$.   We can apply Proposition \ref{cr: corollary}.

\begin{proposition}
  In the example of \citeasnoun{Rubinstein93},
if $\sigma$ is a best response to $\tau_{\Ahat}$, then 
  $(\sigma,\tau_{\Ahat})$ is a Nash equilibrium of the algorithm game, which emulates $(\sigma^R,y^R(\sigma^R,p))$.
\end{proposition}

\section{Conclusion} 

This paper has applied the framework of PAC learnability to describe the performance of algorithms in a strategic setting. We show that as long as some initial set of classifiers satisfy weak learnability, an algorithm can be specified which ensures the receiver takes an optimal response to the sender's action. 
As noted by \citeasnoun{Rubinstein93}, this need not be the case when the receiver's behavior follows from the \textit{optimally} chosen single-threshold classifier given the sender's strategy. However, being able to combine classifiers is enough to overcome this limitation, even if it only remains possible to find the ``best'' classifier from within this limited class.  

Our general analysis has focused on settings featuring strategic inference---based on the observed action of the strategic sender, a rational receiver would update beliefs about an underlying state (thus influencing the optimal response). This adds a complicating feature that the ex-post optimal action is only observed with noise. Yet because this noise diminishes with the size of the sample, we are still able to show this presents no added difficulty (thanks to results from large deviations theory). We briefly mention that if the amount of label noise were bounded away from zero, then our approach need not be successful (a well-known issue with Boosting algorithms). While a technical contribution, it is one that is necessary due to the uncertainty inherent in our applications of interest.

We have sought to articulate the following tradeoff in the design of statistical algorithms to mimic rationality: on the one hand, simply fitting a single-threshold classifier to data will fall short of rational play and be exploited. On the other hand, it may not be clear why this is the end of the story. By adding the ability to fit classifiers repeatedly and combining them in particular ways, we show how the rational benchmark can be restored. Here, we have taken as a black box the ability to fit these classifiers. But given this, our algorithm specifies exactly how to put these fitted classifiers together in order to construct one which can mimic rationality arbitrarily well. 

We have focused on a simple yet general setting where the comparison to the rational benchmark is most transparent. Still, we believe that many concerns highlighted by the machine learning literature regarding the design of algorithms can speak to issues of interest to economic theorists. Given how productive the machine learning literature has been in terms of designing algorithms for the purposes of classification, we hope that our work will inspire further analysis of how these algorithms behave in strategic settings.

\appendix
\footnotesize

\newpage

\section{Weak Learnability Proofs} \label{Proof of Proposition prop: hyperplaneweaklearn}

The proof of \ref{prop: hyperplaneweaklearn} uses the following Lemma: 

\begin{lemma}  \label{lemm: duality}
Let $\mathcal{H}$ be an arbitrary hypothesis class with the property that for every $h \in \mathcal{H}$ and every permutation $\pi : A \rightarrow A$, the composition $\pi \circ h$ is contained in $\mathcal{H}$. Then this hypothesis class can do at least as well as a uniform random guesser. 
\end{lemma}

\begin{proof} 
Let $\Pi$ be the set of all possible permutations on $A$, noting that $\abs{\Pi}= k!$. Fix an arbitrary classifier $h \in \mathcal{H}$, and define $h^{\pi} = \pi \circ h$. Let $c_{j,y}$ be the cost of assigning label $y$ to price $p_{j}$. Define 

\begin{equation*} 
\sum_{\pi \in \Pi} c_{j,h^{\pi}(p_{j})} = \overline{c}_{j}.
\end{equation*}

\noindent In particular, note that this is invariant to the true label of $j$. As a result, the random guesser's expected payoff on observation $j$ is is $\overline{c}_{j}/k!$. To see this, note that $h(p_{j})$ gives some fixed guess regarding the label of price $p_{j}$. Then randomizing over permutations is equivalent to randomizing over labels, as there are an equal number of permutations which flip the label according to $h(p_{j})$ and every other label. 

We therefore obtain the following matrix equation, for an arbitrary $\rho \in (0, \infty)$, where the number of columns is $k!$ and the number of rows is the number of possible prices. 

\begin{equation*} 
\left(
     \begin{array}{c|cc|c}
       &  &  &  \\
   c_{j,h(p)}  & \cdots & \cdots & c_{j,h^{\pi}(p)} \\
   - \overline{c}_{j}/k!  &  &  &  - \overline{c}_{j}/k! \\ &  &  &
     \end{array} \right)
 \cdot \begin{pmatrix}  \frac{\rho}{k!} \\ \vdots \\ \frac{\rho}{k!}  \end{pmatrix} = \mathbf{0}
\end{equation*}

Also note that: 

\begin{equation*} 
(1/\rho, \cdots, 1/\rho) \cdot \begin{pmatrix}  \frac{\rho}{k!} \\ \vdots \\ \frac{\rho}{k!}  \end{pmatrix} = 1
\end{equation*}

So as long as $\rho > 0$, by the theorem of the alternative, we therefore cannot have that a vector $\mathbf{x}$ exists with: 

\begin{equation*} 
    \left(
     \begin{array}{cccc}
  & c_{j,h(p)}  - \overline{c}_{j} / k ! & \\ \hline 
 \vdots  &  \vdots  &  \\  \vdots &     \vdots &   \\ \hline    &  c_{j,h^{\pi}(p)} - \overline{c}_{j} / k ! & \end{array} \right)
     \cdot \mathbf{x} \geq \begin{pmatrix} \frac{1}{\rho} \\ \vdots \\ \frac{1}{\rho} \end{pmatrix}.
\end{equation*}

Let $D(p)$ be an arbitrary distribution. Since $\sum_{p \in P} D(p)=1$, this implies we can find some $\pi$ such that: 

\begin{equation*} 
\left(\sum_{p_{j} \in P} D(p_{j}) (c_{j, h^{\pi}(p_{j})} - \frac{ \overline{c}_{j}}{k!}) \right) < \frac{1}{\rho}.
\end{equation*}

Taking $\rho \rightarrow \infty$ and rearranging gives: 

\begin{equation*} 
\left(\E_{p \sim D}[ c_{j, h^{\pi}(p_{j})} ] \right) \leq \E_{j \sim D} \left[\frac{ \overline{c}_{j}}{k!} \right]
\end{equation*}

Recalling again that the right hand side of this inequality is the payoff of the random guesser, we have shown that for every possible distribution over prices, we can find some permutation which delivers a cost bounded above by the random guesser. This proves the Lemma. \end{proof} 

\begin{proof}[Proof of Proposition \ref{prop: hyperplaneweaklearn}] 
Let $\mathcal{H}$ be the set of hyperplane classifiers. We prove this by contradiction. If there were no universal lower bound on the error, then we would have, for all $\rho$, a distribution $D_{\rho}$ and cost $c_{j,y}^{\rho}$ (without loss normalized to be on the unit sphere themselves) with the property that: 

\begin{equation*} 
\max_{h \in \mathcal{H}}  \sum_{p \in P} D_{\rho}(p) c_{j,h(p)}^{\rho} < U_{c}^{\rho}, 
\end{equation*}

\noindent where $U_{c}^{\rho}$ is the payoff of the uniform random guesser who is correct with added probability $\rho$. Taking $\rho \rightarrow 0$ and passing to a subsequence if necessary, compactness of the unit sphere implies that we can find a distribution $D^{*}$ and cost function $c^{*}$ such that: 

\begin{equation*} 
\max_{h \in \mathcal{H}}  \sum_{p \in P} D^{*}(p) c_{j,h(p)}^{*} = U_{c}^{0},
\end{equation*}

\noindent where we note by Lemma \ref{lemm: duality} that at least this bound can be obtained by permutation the labels if necessary. We will arrive at a contradiction by exhibiting a single-hyerplane classifier that achieves a strictly better accuracy, given $D^{*}$. Note that $\mathcal{H}$ contains the set of ``trivial'' classifiers, which give all menus the same label. Also note that the only non-trivial case to consider is when there are at least two prices in the support of $D^{*}$; if there were only one price, then simply choosing the prediction corresponding to the label on that price would yield a perfect fit. Since, by assumption, no classifier does better than random guessing, it must be the case in particular that each trivial classifier cannot exceed the random-guess bound. On the other hand, by our previous result, we know there \emph{does} exist a trivial classifier which achieves at least this bound, for \emph{any} $D$ supported on $P$.

Let $P=\{p_{1}, \ldots, p_{k}\}$ be the set of prices supporting $D^{*}$, and let $\tilde{p} \in P$ be a price in that is also an extreme point of the convex hull of $P$. Without loss of generality, assume that $\tilde{p}$ is nontrivial, in the sense that it does not give the same cost to all labels. Note that indeed, this is without loss, since for any such price, the choice of classification is irrelevant.\footnote{If all prices are trivial, then we will achieve a contradiction, because that implies that the classifier does do at least as well as the edge-over-random guesser, since all classifiers achieve the same payoff.} Note that $\tilde{p}$ is not in the convex hull of $P \backslash\{ \tilde{p} \}$. Therefore, by the separating hyperplane theorem, we can find an $h \in \mathcal{H}$ which (strictly) separates $\tilde{p}$ from $P \backslash \{\tilde{p}\}$. Denote such a hyperplane by $h^{*}$, and note that the set of hyperplane classifiers contains classifiers which assign \emph{any} two labels (possibly the same label) to prices depending on which side of $h^{*}$ they lie on. 

Also note that, again by our previous result, a trivial classifier supported on $P \backslash \{ \tilde{p}\}$ can achieve the random guess guaranatee if $p$ is distributed according to the conditional distribution on this set. In other words, our prior lemma implies that there exists $y^{*} \in A$ such that: 

\begin{equation*} 
\sum_{p_{j} \in P \backslash \tilde{p}} \frac{D^{*}(p_{j})}{\sum_{q \in P \backslash \tilde{P}} D^{*}(q)} c_{j,y^{*}}^{*} = U_{c^{*}}^{0}.
\end{equation*}

On the other hand, a classifier which separates $p_{\tilde{j}}$ from the other prices can fit $p_{\tilde{j}}$ perfectly. Thus we must have

\begin{equation*} 
c_{\tilde{j}, y_{\tilde{j}}} < \E_{\hat{y} \sim \text{Unif}}[ c_{\tilde{j}, \hat{y}}].
\end{equation*}

So consider the hyperplane classifier which predicts $\tilde{y}$ for $\tilde{p}$, and $y^{*}$ for $p \in P \backslash \{\tilde{p}\}$, i.e., depending on which side of $h^{*}$ they are on (acknowledging that this may be a trivial classifier). Denote the resulting classifier by $h$. For this single-hyperplane classifier, we have

\begin{equation*} 
 \sum_{p_{j} \in P} D^{*}(p_{j}) c_{j,h(p_{j})}  = D^{*}(p_{\tilde{j}}) c_{\tilde{j}, y_{\tilde{j}}} +\left( \sum_{q \in P \backslash \{p_{\tilde{j}}\}} D^{*}(q)\right) \sum_{p_{k} \in P \backslash \{p_{\tilde{j}}\}} \frac{D^{*}(p_{k})}{\sum_{q \in P \backslash \{p_{\tilde{j}} \}} D^{*}(q)} c_{k,y^{*}} > U_{c^{*}}^{0},
\end{equation*}

\noindent where the inequality holds since the single-threshold classifier does strictly better on some non-trivial price, and as well on all other prices. This completes the proof. 
 \end{proof}

\section{Specifying the Algorithm Parameters and the Proof of Proposition \ref{pr: convergence}.}
\label{Proof of Proposition pr: convergence}

\subsection{Convergence of $\tau_A$}
\subsubsection{The $\abs{A}=2$ case}
We replicate the proof in \citeasnoun{SchapireandFreund12} for reference.
Define
\[
F_t(p)=\sum_{k=1}^t\alpha_kh_k(p).
\]
Following the same recursive process described in \citeasnoun{SchapireandFreund12}, we have
\begin{equation}
d_{t+1}(p) =\frac{
d_1(p) \exp \left(-y(\sigma,p)\sum_{k=1}^t \alpha_k h_k(p)\right)
}{
\prod_{k=1}^tZ_k
}
=\frac{
d_1(p)\exp (-y(\sigma,p)F_t(p))
}{
\prod_{k=1}^tZ_k }.
\label{eq: Dt expression}
\end{equation}

Following \citeasnoun{SchapireandFreund12}, we can show that
\[
\Prob\left(H_t(p)\ne y(\sigma,p)\right)
=\Expect\sum_{p} d_1(p){\mathbf 1}(H_t(p)\ne y(\sigma,p))
\le\Expect \sum_{p} d_1(p)\exp (-y(\sigma,p)F_t(p)),
\]
and
\[
\Prob ( H_t(p)\ne y(\sigma,p))
=\Expect\prod_{k=1}^t Z_k.
\]
Note
\[
  Z_k =\sum_p d_k(p) \exp\left(-y(\sigma,p) \alpha_k h_k(p)\right).
\]
The rest of the proof follows from \citeasnoun{SchapireandFreund12}, which we copy here for later reference.
\begin{eqnarray*}
  Z_t & = & \sum_p d_t(p) \exp\left(-y(\sigma,p) \alpha_t h_t(p)\right) \\
      & = & \sum_{y(\sigma,p) h_t(p)=1} d_t(p) \exp\left(-\alpha_t \right) +
            \sum_{y(\sigma,p) h_t(p)=-1} d_t(p) \exp\left(-\alpha_t \right) \\
      & = & e^{-\alpha_t} (1-\epsilon_t) +e^{\alpha_t} \epsilon_t \\
      &= &  e^{-\alpha_t} \left( \frac{1}{2}+\gamma_t\right) +e^{\alpha_t}
           \left(\frac{1}{2}-\gamma_t  \right) \\
      & = & \sqrt{1-4\gamma^2_t }
\end{eqnarray*}
where
\[
\gamma_t =\frac{1}{2}-\epsilon_t.
\]
By weak learnability, we know that $\gamma_t$ is uniformly bounded away from 0:  $\exists\gamma>0$ such that
\[
  \gamma_t\ge\gamma \qquad\forall t\ge 1.
\]
Recall that the maximum number of the elements in the support of $\sigma$ is $N$.  Thus,
\[
  d_{t+1}(p)=d_1(p)\prod_{k=1}^t \sqrt{1-4\gamma^2_t } \le
  \frac{1}{N} \left( 1-4 \gamma^2\right)^{\frac{t}{2}}
  \le \frac{1}{N} e^{-2\gamma^2 t}
\]
where the right hand side converges to 0 at the exponential rate uniformly over $p$.

\subsection{The $\abs{A} > 2$ case} 
The specification of the algorithm can be found in \citeasnoun{MukherjeeSchapire2013}. The proof provided below fills in some details to show that convergence holds in a self-contained way.  

First,  initialize $F_{y}^{0}(x_{i})=0$. 

\begin{itemize} 
\item From previous stage, take $F_{y}^{t}$. 

\item At stage $t$, find the $h \in \mathcal{H}$ solving: 

\begin{equation*} 
\min_{h \in \mathcal{H}} \frac{1}{m} \sum_{i=1}^{m} \mathbf{1}[h_{t}(x_{i})=y_{i}]\left((e^{-\eta}-1) \sum_{\tilde{y} \neq y_{i}} e^{\eta(F_{\tilde{y}}^{t-1}-F_{y_{i}}^{t-1})}\right)+ \mathbf{1}[h_{t}(x_{i}) \neq y_{i}](e^{\eta}-1)e^{\eta (F_{h_{t}(x_{i})}^{t-1}-F_{y}^{t-1}(x_{i}))}. 
\end{equation*}

\item Define $F_{y}^{t}(x_{i}) = \sum_{s=1}^{t} \mathbf{1}[h_{t}(x_{i})=y]$. 

\end{itemize}

The final prediction is $H_{t}(x_{i}) = \argmax_{\tilde{y}} \sum_{t=1}^{T} \mathbf{1}[h_{t}(x_{i}) = \tilde{y}]. $

The weak learnability condition says that the hypothesis class can outperform a random guesser that does better than some $\gamma$, where we allow for a potentially asymmetric cost of making different errors. 

We now show convergence to the rational rule:

\textbf{Step 1: Bounding The Mistakes}: This step is as previous. We have

\begin{equation*} 
\sum_{i=1}^{m} \mathbf{1}[H_{t}(x_{i}) \neq y_{i}] \leq \sum_{i=1}^{m} \sum_{\tilde{y} \neq y_{i}} e^{\eta( F_{\tilde{y}}^{t}(x_{i})- F_{y_{i}}^{t}(x_{i}))}. 
\end{equation*}

Indeed, the exponential is positive, so this inequality holds when $y_{i}$ is labelled correctly, and if the label is incorrect, then that means that some $\tilde{y}_{i}$ satisfies $F_{\tilde{y}_{i}}^{t}(x_{i}) > F_{y_{i}}^{t}(x_{i})$. Since all exponential terms are positive, and furthermore the exponent is positive if $x_{i}$ is labelled incorrectly, meaning the right hand side is greater than 1 if mislabeled.

\textbf{Step 2: Recursive Formulation of the Loss} We now show that the right hand side goes to 0 at an exponential rate. We define the loss function to be: 

\begin{equation*} 
L_{t}(x_{i}) = \sum_{\tilde{y} \neq y_{i}} e^{\eta(F_{\tilde{y}}^{t}(x_{i}) - F_{y_{i}}^{t}(x_{i}))}, \tilde{L_{t}} = \frac{1}{m} \sum_{i=1}^{m} L_{t}(x_{i}). 
\end{equation*}

We first express $\tilde{L}_{t+1}$ as a function of $\tilde{L}_{t}$. Note that $F_{\tilde{y}}^{t+1}(x_{i}) = F_{\tilde{y}}^{t}(x_{i})$ for all $\tilde{y} \neq h_{t}(x_{i})$, and $F_{\tilde{y}}^{t+1}(x_{i}) = F_{\tilde{y}}^{t}(x_{i}) + 1$ for $\tilde{y} = h_{t}(x_{i})$. The loss from a given $x_{i}$ changes depending on whether or not it is correctly classified. For any observation that is classified correctly at the $t+1$th stage, we multiply that observation's loss by a factor of $e^{-\eta}$. On the other hand, for any observation that is classified incorrectly as $\tilde{y}$, we \emph{add} the following: 

\begin{equation*} 
e^{\eta(F_{\tilde{y}}^{t}(x_{i}) - F_{y_{i}}^{t}(x_{i}))}(e^{\eta}-1).
\end{equation*} 

So: 

\begin{equation*} 
\tilde{L}_{t+1} = \frac{1}{m} \left( \sum_{i: h_{t+1}(x_{i})= y_{i}} e^{-\eta} L_{t}(x_{i})+ \sum_{i : h_{t+1}(x_{i}) \neq y_{i}} \left(L_{t}(x_{i}) + e^{\eta(F_{h_{t+1}(x_{i})}^{t}(x_{i}) - F_{y_{i}}^{t}(x_{i}))}(e^{\eta}-1) \right) \right).
\end{equation*}

Note that if we subtract $\tilde{L}_{t}$ from both sides, and substitute in for $L_{t}(x_{i})$ above, we obtain: 

\begin{equation*} 
\tilde{L}_{t+1} -\tilde{L}_{t}= \frac{1}{m} \left( \sum_{i: h_{t+1}(x_{i})= y_{i}} (e^{-\eta} -1) \sum_{\tilde{y} \neq y_{i}} e^{\eta(F_{\tilde{y}}^{t}(x_{i}) - F_{y_{i}}^{t}(x_{i}))} + \sum_{i : h_{t+1}(x_{i}) \neq y_{i}} e^{\eta(F_{h_{t+1}(x_{i})}^{t}(x_{i}) - F_{y_{i}}^{t}(x_{i}))}(e^{\eta}-1) \right).
\end{equation*}

\textbf{Step 3: Weak Learnability} By the above, $h_{t+1}$ is chosen to solve: 

\begin{equation*} 
\min_{h \in \mathcal{H}} \frac{1}{m} \sum_{i=1}^{m} \mathbf{1}[h (x_{i})=y_{i}]\left((e^{-\eta}-1) \sum_{\tilde{y} \neq y_{i}} e^{\eta(F_{\tilde{y}}^{t}(x_{i})-F_{y_{i}}^{t}(x_{i}))}\right)+ \mathbf{1}[h(x_{i}) \neq y_{i}](e^{\eta}-1)e^{\eta (F_{h(x_{i})}^{t}(x_{i})-F_{y}^{t}(x_{i}))}. 
\end{equation*} 

In fact, using the previous step, we see that this can equivalently be expressed as $\tilde{L}_{t+1} - \tilde{L}_{t}$. On the other hand, someone who is random guessing, but is correct with extra probability $\gamma$, will be correct with probability $\frac{1-\gamma}{k} + \gamma$, and guess an incorrect label $\tilde{y}$ with probability $\frac{1- \gamma}{k}$.  Furthermore, the hypothesis class ensures a weakly lower error (as measured by this cost) than the random guessing. Hence this expression is bounded above by: 

\begin{equation*} 
\frac{1}{m}\sum_{i=1}^{m} \left(( \frac{1- \gamma}{k} + \gamma)  (e^{-\eta} -1)L_{t}(x_{i}) + \frac{1- \gamma}{k}\sum_{\tilde{y} \neq y_{i}}(e^{\eta}-1) e^{\eta (F_{\tilde{y}}^{t}(x_{i}) - F_{y}(x_{i}))} \right)
\end{equation*} 

Again substituting in for $L_{t}(x_{i})$ and rearranging, we obtain: 

\begin{equation*} 
\left(( \frac{1- \gamma}{k} + \gamma)(e^{-\eta} - 1) + \frac{1- \gamma}{k} (e^{\eta} - 1) \right)  \tilde{L}_{t}. 
\end{equation*}

Putting this together, we have this is an upper bound of $\tilde{L}_{t+1} - \tilde{L}_{t}$, and therefore: 

\begin{equation*} 
\tilde{L}_{t+1} \leq \left(1 + \left(( \frac{1- \gamma}{k} + \gamma)(e^{-\eta} - 1) + \frac{1- \gamma}{k} (e^{\eta} - 1) \right) \right)  \tilde{L}_{t}. 
\end{equation*}

\textbf{Step 4: Specifying $\eta$} We are done if we can ensure $\tilde{L}_{t}\rightarrow 0$ as $t \rightarrow \infty$, since Step 1 shows that this implies that the number of misclassifications approaches 0 as well. To complete the argument, we must specify an $\eta$  which delivers the exponential convergence. However, first note that if $\eta=0$, the coefficient on $\tilde{L}_{t}$ in the previous inequality is 1, and the derivative with respect to $\eta$ is $- \gamma$ at 0, so that this expression is less than 1, for some $\eta>0$. Setting $\eta=\log(1+\gamma)$, the above coefficient on $\tilde{L}_{t}$ reduces to: 

\begin{equation*} 
\overbrace{1+ \left( ( \frac{1- \gamma}{k} + \gamma)(\frac{1}{1+ \gamma}-1) + \frac{1- \gamma}{k}\gamma \right)}^{z_{k}(\gamma)}.
\end{equation*}  

Note that $z_{k}(\gamma)$ is bounded above by $\tilde{z}(\gamma)=e^{ - \gamma^{2}/2}$. Indeed, this expression is decreasing in $k$, with $z_{k}(0)=1=\tilde{z}(0)$, and $z_{2}(\gamma) = 1 - \frac{\gamma^{2}}{2} < e^{-\gamma^{2}/2}=\tilde{z}(\gamma)$. Since $\tilde{L}_{0}= (k-1)$, we therefore have that: 

\begin{equation*} 
\tilde{L}_{t} \leq (k-1)e^{- \gamma t^{2}/2},
\end{equation*} 

as desired.

\subsection{Convergence of $\tau_{\Ahat}$ }

Under the assumption that $y^R(\sigma,p)$ is a strict best response,
\[
\lim_{t\rightarrow\infty}\yhat_t(p)= y^R(\sigma,p) 
\]
almost surely.   Since $\yhat_t(p)$ satisfies the uniform LDP,
$\forall\epsilon>0$, $\exists\rho(\epsilon,\sigma)>0$ and $T(\epsilon,\sigma)$ such that
\[
  \Prob\left(
\exists t\ge T(\epsilon,\sigma), \yhat_t(p)\ne y^R(\sigma,p) 
    \right) \le e^{-t \rho(\epsilon,\sigma)}.
\]
Since the support of $\sigma$ contains a finite number of $p$, the empirical 
the multinomial probability distribution over $\theta$.

Let $\pihat_t(\theta : p)$ be the empirical probability distribution over $\Theta$ following $t$ rounds of observations.
By the law of large numbers, $\pihat_t(\theta : p)\rightarrow \pi(\theta :p)$ computed via Bayes rule from the prior distribution over $\theta$ and $\sigma$.
Write $\Theta=(\theta_1,\ldots,\theta_{\abs{\Theta}})$.    
Given ${\bf\epsilon}=(\epsilon,\ldots,\epsilon)\in {\mathbb R}^{\abs{\Theta}}$, the rate function of the multinomial distribution is
\[
\sum_{i=1}^{\abs{\Theta}}\epsilon\log \frac{\epsilon}{p(\theta)}
\]
where $p(\theta)$ is the probability that $\theta$ is realized.   Since $\sum_{\theta}p(\theta)=1$, 
\[
  \sum_{i=1}^{\abs{\Theta}}\epsilon\log \frac{\epsilon}{p(\theta)} \ge
  \prod_{i=1}^{\abs{\Theta}}\epsilon\log\frac{\epsilon}{1/\abs{\Theta}}
  =  \prod_{i=1}^{\abs{\Theta}}\epsilon\log \epsilon\abs{\Theta} >0.
\]
Note that the right hand side is independent of $\sigma$, which is the rate function of the uniform distribution over $\Theta$.
Thus, we can choose $\rho(\epsilon)\le \rho(\epsilon,\sigma)$ uniformly over $\sigma$, which is strictly increasing with respect to $\epsilon>0$.    We choose $T(\epsilon)$ independently of $\sigma$ as well.

Define an event
\[
  {\mathcal L}=\left\{
\yhat_t(p) =y^R(\sigma,p) \qquad\forall t\ge T(\epsilon)
    \right\}
\]
We know that
\[
  \Prob({\mathcal L})\ge 1-e^{-t\rho(\epsilon)}.
\]
Fix $t > T(\epsilon)$.   We have
\begin{eqnarray*}
&&\Prob\left( \tau_{\Ahat}(D_t)(p)\ne y^R(\sigma,p)\right)\\
  & = & \Prob\left( \tau_{\Ahat}(D_t)(p)\ne y^R(\sigma,p) : {\mathcal L}\right) \Prob({\mathcal L}) +\Prob\left( \tau_{\Ahat}(D_t)(p)\ne y^R(\sigma,p) : {\mathcal L}^c\right) \Prob({\mathcal L}^c)  \\
  & \le &
          \Prob\left( \tau_{\Ahat}(D_t)(p)\ne y^R(\sigma,p) : {\mathcal L}\right)+ \Prob({\mathcal L}^c)  \\
  & \le &
 \Prob\left( \tau_{\Ahat}(D_t)(p)\ne y^R(\sigma,p) : {\mathcal L}\right)+ e^{-t\rho(\epsilon)}.
\end{eqnarray*}
Following the same logic as in the proof of Proposition \ref{pr: convergence}, we can show that $\exists \gamma(G)>0$ such that
\begin{equation}
\Zhat_t \le 1-\gamma(G) \qquad \forall t\ge 1  \label{eq: second bound}
\end{equation}
under $\tau_{\Ahat}$.

Recall that
\[
F_a(p)= \sum_{s=1}^t\alpha_s{\mathbf 1}(h_s(p)=a).
\]
Similarly, we define
\[
\Fhat_a(p)= \sum_{s=1}^t\alphahat_s{\mathbf 1}(h_s(p)=a).
\]
Following the same logic as in the proof of Proposition \ref{pr: convergence}, we know that if $\tau_{\Ahat}(D_t)(p)\ne y^R(p)$,
\[
\Fhat_{y^R(\sigma,p)}(p) +\sum_{a\ne y^R(\sigma,p)}\Fhat_a(p) >0.
\]
Thus,
\begin{eqnarray*}
{\mathbf 1}( \tau_{\Ahat}(D_t)(p)\ne y^R(\sigma,p)) 
  & \le & {\mathbf 1}
    \left( \Fhat_{y^R(\sigma,p)}(p) +\sum_{a\ne y^R(\sigma,p)}\Fhat_a(p)   \right) \\
  & \le & \exp\left(\Fhat_{y^R(\sigma,p)}(p) +\sum_{a\ne y^R(\sigma,p)}\Fhat_a(p)   \right).
\end{eqnarray*}
Conditioned on event ${\mathcal L}$,
\[
\yhat_t(p)=y^R(\sigma,p) \qquad\forall t\ge T(\epsilon).
\]
We can write for $t\ge T(\epsilon)$,
\begin{eqnarray*}
  d_{t+1}(p) & = & \frac{\dhat_t(p) \exp(\alpha_t( {\mathbf 1}(h_t(p)\ne \yhat_t(p))-  {\mathbf 1}(h_t(p)= \yhat_t(p))))
                   }{ \Zhat_t   }   \\
  & = & \frac{\dhat_t(p) \exp(\alpha_t( {\mathbf 1}(h_t(p)\ne y^R(\sigma,p))-  {\mathbf 1}(h_t(p)= y^R(\sigma,p))))
                   }{ \Zhat_t   }   \\
  & = & \frac{d_{T(\epsilon)}(p) \exp(\sum_{s=T(\epsilon)}^t\alpha_s( {\mathbf 1}(h_s(p)\ne y^R(\sigma,p))-  {\mathbf 1}(h_s(p)= y^R(\sigma,p))))
                   }{\prod_{s=T(\epsilon)}^t \Zhat_t   }.
\end{eqnarray*}
Thus,
\begin{eqnarray*}
  &&\prod_{s=T(\epsilon)}^t \Zhat_t \\
  &=&\sum_p d_{T(\epsilon)}(p)
      \exp\left[\sum_{s=T(\epsilon)}^t\alpha_s( {\mathbf 1}(h_s(p)\ne y^R(\sigma,p))-  {\mathbf 1}(h_s(p)= y^R(\sigma,p)))\right] \\
    &\ge&\left(\min_{p\in{\mathcal P}(\sigma)}d_{T(\epsilon)}(p)\right) \sum_p \exp\left[\sum_{s=T(\epsilon)}^t\alpha_s( {\mathbf 1}(h_s(p)\ne y^R(\sigma,p))-  {\mathbf 1}(h_s(p)= y^R(\sigma,p)))\right].
\end{eqnarray*}
Since $d_1(p)$ is the uniform distribution over ${\mathcal P}(\sigma)$,
\[
\min_{p\in{\mathcal P}(\sigma)}d_{T(\epsilon)}(p) >0.
\]
We can write 
\begin{eqnarray*}
  && \prod_{s=1}^t \Zhat_t=\prod_{s=T(\epsilon)}^t\Zhat_t \prod_{s=1}^{T(\epsilon)-1}\Zhat_t  \\
  &\ge &   \left(\min_{p\in{\mathcal P}(\sigma)}d_{T(\epsilon)}(p)\right) \sum_{p}\exp(\sum_{s=T(\epsilon)}^t\alphahat_s( {\mathbf 1}(h_s(p)\ne y^R(\sigma,p))-  {\mathbf 1}(h_s(p)= y^R(\sigma,p)))) \prod_{s=1}^{T(\epsilon)-1}\Zhat_t \\
  &= & \frac{  \left(\min_{p\in{\mathcal P}(\sigma)}d_{T(\epsilon)}(p)\right)  \prod_{s=1}^{T(\epsilon)-1}\Zhat_t}{ \sum_{p}\exp\left[\sum_{s=1}^{T(\epsilon)-1}\alphahat_s( {\mathbf 1}(h_s(p)\ne y^R(\sigma,p))-  {\mathbf 1}(h_s(p)= y^R(\sigma,p)))\right]} \\
  && \ \ \ \ \ \ \ \ \ \ \ \ \ \times
  \sum_{p}\exp\left[\sum_{s=1}^t\alphahat_s( {\mathbf 1}(h_s(p)\ne y^R(\sigma,p))-  {\mathbf 1}(h_s(p)= y^R(\sigma,p)))\right]
\end{eqnarray*}
over ${\mathcal L}$.
Define
\[
M(\epsilon) =\frac{  \left(\min_{p\in{\mathcal P}(\sigma)}d_{T(\epsilon)}(p)\right)  \prod_{s=1}^{T(\epsilon)-1}\Zhat_t}{ \sum_{p}\exp(\sum_{s=1}^{T(\epsilon)-1}\alphahat_s( {\mathbf 1}(h_s(p)\ne y^R(\sigma,p))-  {\mathbf 1}(h_s(p)= y^R(\sigma,p))))}
\]
which is bounded away from 0.

Recall that
\begin{eqnarray*}
  &&  \Prob( \tau_{\Ahat}(D_t)(p) \ne y^R(\sigma,p) ) \\
  & \le & \sum_p d_1(p) \exp(\sum_{s=1}^t\alphahat_s( {\mathbf 1}(h_s(p)\ne y^R(\sigma,p))-  {\mathbf 1}(h_s(p)= y^R(\sigma,p)))) \\
& \le & \frac{\prod_{s=1}^t \Zhat_t}{M(\epsilon)} 
\le \frac{ (1-\gamma(G))^t}{M(\epsilon)} \le \frac{e^{-t \gamma(G)}}{M(\epsilon)}.
\end{eqnarray*}

Combining the probabilities over ${\mathcal L}$ and ${\mathcal L}^c$, we have that $\forall\epsilon$, $\forall\sigma\in\Sigma^G\subset\Sigma$, $\exists T(\epsilon)$, $\rho(\epsilon)$ and $\gamma(G)$ such that
\[
  \Prob\left( \exists t\ge T(\epsilon), \ \tau_{\Ahat}(D_t)(p)\ne y^R(\sigma,p)   \right)
  \le \frac{e^{-t \gamma(G)}}{M(\epsilon)} +e^{-t \rho(\epsilon)}.
\]
We can choose $T>T(\epsilon)$ and ${\overline\rho}$ such that $\forall t\ge T$,
\[\frac{e^{-t \gamma(G)}}{M(\epsilon)} +e^{-t \rho(\epsilon)}\le e^{-{\overline\rho}t}
  \]
which proves the proposition.

\section{Proofs for Section \ref{sect: whyalg}}

\begin{proof}[Proof of Proposition \ref{prop: atleastsingle}]
Concave differences implies that the set

\begin{equation*} 
K=\{p : u(\theta, p, 1) \geq u(\theta, p, 0)\}
\end{equation*} 

\noindent is a convex set; if $u(\theta, p_{i}, 1) - u(\theta, p_{i}, 0) \geq 0$ for $i=1,2$, then the same conclusion holds for $\alpha p_{1} + (1- \alpha)p_{2}$ for all $\alpha \in [0,1]$. Therefore, given any $p^{*}$ on the boundary, the supporting hyperplane theorem implies that we can find a linear hyperplane $(\lambda, \omega)$ tangent to this set at $p^{*}$. 

Suppose the algorithm designer prescribes that the receiver choose $a=1$ at any menu $p$ such that $\lambda \cdot p \leq \omega$ and $a=0$ otherwise. Note that having the receiver choose $a=1$ therefore requires choosing $p$ where the sender would rather the receiver choose action $a=0$, by definition of $K$. Therefore, the strategic player cannot do any better than choosing $\sigma(p  \mid \theta)$ which is a point mass at $p^{*}$. 
\end{proof}

\begin{proof}[Proof of Proposition \ref{thm:fullsurplus}]
  The ideas in this proof are largely borrowed from \citeasnoun{Rubinstein93}, accommodating two additional features of our enviroment: (a) need to infer the strategy from observed data and (b) the generalized setting, but we provide the proof for completeness. We construct a strategy $\sigma^*$ for the sender that generates higher payoff than the equilibrium strategy $\sigma^R$, thus deriving the contradiction that $\sigma^R$ is a best response to $\tau$ in the long run. More precisely, define $(p_{\theta}^{*}, a(\theta))$ to be the sender payoff-maximizing strategy.  We show that the sender can induce the receiver to choose $a(\theta) \neq y_{R}(p_{\theta}^{*})$. 

Fix $\epsilon> 0$ small, and without loss suppose $\abs{\Theta}=2$. First suppose $v(\theta_{L}, p_{L}^{*}, a_{1}) = v(\theta_{L}, p_{L}^{*}, a_{0})$. Let $\tilde{p} \in (p_{L}, p_{H})$ satisfies $v(\theta_{L}, \tilde{p}, 1) < v(\theta_{L}, \tilde{p}, 0)$. (If $p$ is multidimensional, we can take $\tilde{p}$ tp be on the line segment connecting $p_{L}^{*}$ and $p_{H}^{*}$) Set $\eta=  v(\theta_{L}, \tilde{p}, 0)- v(\theta_{L}, \tilde{p}, 1)  >0.$ We then choose $\epsilon,\epsilon_H,\epsilon_L>0$ to satisfy
\begin{equation}
\pi(H)\epsilon_H < \pi(L)\epsilon_L, \label{ineq:first}
\end{equation}
and such that 
\begin{equation} 
  \frac{\epsilon_L}{\epsilon_L+\eta} <\epsilon
  <\frac{\pi(L)\epsilon_L-\pi(H)\epsilon_H}{\pi(L)\epsilon_L}. \label{ineq:second}
\end{equation}

\noindent Under the increasing differences assumption, we can find $p_{i}(\varepsilon_{i})$ such that \begin{equation*} \varepsilon_{i} = v(\theta_{i}, p_{i}(\varepsilon_{i}), 1) - v(\theta_{i}, p_{i}(\varepsilon_{i}), 0). \end{equation*} 

\noindent Consider the following randomized pricing rule $\sigma^*$ of the sender: in state $H$, $\tilde{p}_{H}(\epsilon_H)$ is chosen with probability 1. In state $L$, $p_{L}(\epsilon_{L})$ is chosen with probability $1-\epsilon$ and $\tilde{p}$ with probability $\epsilon$. 

Under this strategy, the optimal response following $\tilde{p}$ is 0, and this does not vanish as all other parameters tend to 0. However, the ex-post optimal decisions are 1 for both $\tilde{p}_{L}(\epsilon_{L})$ and $\tilde{p}_{H}(\epsilon_{H})$. Nevertheless, (\ref{ineq:second}) implies first, the decisionmaker prefers to choose $a=1$ if and only if $\tilde{p}_{L}(\epsilon_{L})$ than choose $a=1$ if and only if $\tilde{p}_{H}(\epsilon_{H})$; and second, that the loss from choosing $a=1$ following $\tilde{p}$ is larger than the loss from choosing $a=0$ at $\tilde{p}_{L}(\epsilon_{L})$. Putting this together, and taking $\epsilon, \epsilon_{L}, \epsilon_{H} \rightarrow 0$ shows this policy approximates the sender's optimum,  as desired. 

The case of $v(\theta_{L}, p_{L}^{*}, a_{1}) > v(\theta_{L}, p_{L}^{*}, a_{0})$ is even more straightforward, since in this case the gain from choosing $a_{1}$ is non-vanishing, meaning that we can set $\varepsilon_{L}=0$.

The verification that the optimal rule converges to this threshold when emerging from data is straightforward; any recursive learning algorithm generates $\{ \phi_t\}$ which converges to $\phi\in \left(v_L-\epsilon_L, \frac{v_H+v_L}{2}\right)$ to emulate the best response of type 1 buyer against $\sigma$.    Thus, the long run average payoff against such algorithm should be bounded from below by ${\mathcal U}_p^*-\epsilon$.
\end{proof}

\section{Proofs for Section \ref{sect: discretization}}
\subsection{Proof of Proposition \ref{pr: Smoothing}}

The proof of the theorem proceeds in the following steps: 

\begin{itemize} 
\item Step 1: Show that the expected value conditional on price, in the image of the sender's possible strategies after applying the augmentation, is uniformly equicontinuous. 

\item Step 2: Show that the same label is applied to $\E[v_{\theta} \mid p+z_{i,\eta}, \sigma, \phi_{\eta}]$ as would be applied to $\E[v_{\theta} \mid p, \sigma]$, with high probability.

\item Step 3: Verify that the change in recommendation due to discarding ``low density prices'' occurs with vanishing probability.
\end{itemize}

\noindent Putting these together shows that the change in the expectation can be made arbitrarily small, as can the probability that small density observations are drawn. The condition that $\sigma$ is either discrete or continuous is stronger than necessary; what is necessary is continuity of the conditional expectation as a function of price, which can be satisfied if the discrete portions and continuous portions are separated, for instance. However, the proposition highlights that we need not restrict the sender's strategy space at all in order for our algorithm to converge. 

The Theorem implies that if the sender were to use an \textit{arbitrary} strategy $\sigma$, the receiver could instead focus on finding a rational response to $\tilde{\sigma}_{\eta}$. Doing so would still lead to PAC learnability of the approximately optimal response to $\sigma$. On the other hand, we can show that the optimal response to $\tilde{\sigma}_{\eta}$ is PAC learnable (unlike, potentially, the optimal response to $\sigma$), and doing the change leads to a negligible impact on the sender's surplus. 

Before presenting the proof, we argue that uniform equicontinuity implies weak learnability. Suppose that $\E[v \mid \sigma, p]-p$ is uniformly equicontinuous (which holds if $\E[v \mid \sigma, p]$ is uniformly equicontinuous). By uniform equicontinuity, we have there exists some $\delta$ such that whenever $\abs{p-p'}< \delta$, we have that \begin{equation*} \abs{\E[v \mid \sigma, p]-\E[v \mid \sigma, p']} < 2\varepsilon, \end{equation*} for any $\sigma$. Suppose we have some price $p$ such that $\E[v \mid \sigma, p] - p > \varepsilon$. Then if $\E[v \mid \sigma, p'] - p' < -\varepsilon$, it follows that $\abs{p-p'} > \delta$. It follows that there can only be at most $\frac{v_{H}- v_{L}}{\delta}$ prices such that $y(\sigma, p)=-y(\sigma,p')$, where $p$ and $p'$ are adjacent (ignoring all prices where $\abs{\E[v \mid \sigma, p] - p} < \varepsilon$, as the classification decision is irrelevant there).

\subsubsection{Step One} We first show that $\E[v_{\theta} \mid \tilde{\sigma}_{\eta}, p]$ is Lipschitz in $p$ uniformly of $\tilde{\sigma}_{\eta}$, noting that we are restricting to prices where $\tilde{\sigma}_{\eta} (p) > \gamma$. Note that:

\begin{equation*} 
\tilde{\sigma}_{\eta}'(p \mid \theta) = \int \phi_{\eta}'(p- \tilde{p}) \sigma(\tilde{p} \mid \theta) d\tilde{p} \leq \max \phi_{\eta}' := \overline{\phi'}.
\end{equation*}

Furthermore, we have: 

\begin{equation*} 
\frac{d}{dp}\Prob_{\tilde{\sigma}_{\eta}}[\theta \mid p]= \frac{\tilde{\sigma}_{\eta}'(p \mid \theta)\Prob[\theta]}{\sum_{\tilde{\theta}} \tilde{\sigma}_{\eta}(p \mid \tilde{\theta}) \Prob[\tilde{\theta}]} - \frac{ \tilde{\sigma}_{\eta}(p \mid \theta)\Prob[\theta](\sum_{\tilde{\theta}}  \sigma_{\eta}'(p \mid \tilde{\theta}) \Prob[\tilde{\theta}])}{{(\sum_{\tilde{\theta}} \tilde{\sigma}_{\eta}(p \mid \tilde{\theta}) \Prob[\tilde{\theta}])^{2}}},
\end{equation*}

so: 

\begin{equation*} 
\abs{\frac{d}{dp}\Prob_{\tilde{\sigma}_{\eta}}[\theta \mid p]} \leq \overline{\phi'} \Prob[\theta] \cdot \left(\frac{1}{\sum_{\tilde{\theta}} \tilde{\sigma}_{\eta}(p \mid \tilde{\theta}) \Prob[\tilde{\theta}]} \right) + \overline{\phi'}\left( \frac{\tilde{\sigma}_{\eta}(p \mid \theta)\Prob[\theta]}{(\sum_{\tilde{\theta}} \tilde{\sigma}_{\eta}(p \mid \tilde{\theta}) \Prob[\tilde{\theta}])^{2}} \right) \leq \overline{\phi'}  \Prob[\theta] \left(\frac{1}{\gamma} + \frac{M(\eta)}{\gamma^{2}} \right),
\end{equation*}

\noindent where $M(\eta)$ is a bound on $\tilde{\sigma}_{\eta}(p \mid \theta) \Prob[\theta]$, which exists since $\sigma$ and $\phi_{\eta}$ have bounded densities. Hence we see that for all $p \neq p^{*}$, the conditional probability is uniformly bounded in $p$, and is hence Lipschitz continuous. Importantly, the bound only depends on $\eta$ and $\gamma$ (and $\Prob[\theta]$), and is therefore uniform over all strategies in the image of the augmentation. Hence we can ensure that Lipschitz continuity is mainted for all prices in the support of $\tilde{\sigma}_{\eta}$. 

In fact, recall that the Lipschitz constant is equal to the $L^{\infty}$ norm of the derivative. Hence Lipschitz continuity depends only on $\gamma$, $M(\eta)$ and $\overline{\phi_{\eta}'}$, meaning that the Lipschitz constant holds uniformly over the image of the distributions emerging under the algorithm.  It follows that the image is uniformly equicontinuous. 

\subsubsection{Step Two}
Note that since $\E[v_{\theta} \mid \sigma, p]$ is continuous on $S=\cup_{\theta}~ \text{  Supp  } ~ \sigma( \cdot \mid \theta)$, $\E[v_{\theta} \mid \sigma, p]$ is uniformly continuous on any compact $K \subset S$. Define: 
\begin{equation*} 
K_{\gamma}= \{p : \sum_{\theta} \sigma(p \mid \theta) \Prob[\theta] \geq \gamma \}.
\end{equation*} Using that mollifiers converge uniformly on compact sets, we have that $\tilde{\sigma}_{\eta} \rightarrow \sigma$ uniformly on $K_{\gamma}$. We therefore have that, for any $\tilde{\varepsilon}$, we can find some $\overline{\eta}$ such that if $\eta < \overline{\eta}$ and $p \in K_{\gamma}$, then $\abs{\tilde{\sigma}_{\eta}(p \mid \theta) - \sigma(p \mid \theta)} < \tilde{\varepsilon}$ for all $\theta$, and $\abs{\sum_{\theta} \tilde{\sigma}_{\eta}(p \mid \theta)\Prob[\theta] - \sum_{\theta} \sigma( p \mid \theta)\Prob[\theta]} < \tilde{\varepsilon}$. 

Furthermore, since $\sigma$ is uniformly continuous on $K_{\gamma}$, we have: 

\begin{equation*} 
\abs{\sigma(p \mid \theta) - \tilde{\sigma}(p' \mid \theta)} = \abs{\int \phi_{\eta}(p'- \tilde{p}) (\sigma(p \mid \theta) - \sigma(\tilde{p} \mid \theta )) d \tilde{p}} \leq \tilde{\varepsilon},
\end{equation*}

using the uniform continuity of $\sigma$ on $K_{\gamma}$. 

So for any $p \in K_{\gamma}$, and $\eta$ sufficiently small, we have (letting $\overline{v} = \max_{\theta} v_{\theta}$):
\begin{align*} 
 \abs{\E[v_{\theta} \mid \sigma, p]-\E[v_{\theta} \mid \tilde{\sigma}_{\eta}, p'] } &= \abs{ \frac{\sum_{\theta} v_{\theta} \sigma(p \mid \theta) \Prob[ \theta] \sum_{\tilde{\theta}} \tilde{\sigma}_{\eta}(p' \mid \tilde{\theta}) \Prob[\tilde{\theta}] -\sum_{\theta} v_{\theta} \tilde{\sigma}_{\eta}(p' \mid \theta) \Prob[ \theta] \sum_{\tilde{\theta}} \sigma(p \mid \tilde{\theta}) \Prob[\tilde{\theta}]   }{\left(\sum_{\theta} \sigma(p \mid \theta) \Prob[\theta] \right) \left( \sum_{\theta} \tilde{\sigma}_{\eta}(p' \mid \theta) \Prob[\theta] \right)}} \\ & \leq \frac{1}{\sigma(p) \cdot (\gamma - \tilde{\varepsilon})} \biggl\lvert \sum_{\theta} v_{\theta} (\sigma(p \mid \theta) - \tilde{\sigma}_{\eta}(p' \mid \theta)) \Prob[\theta] \sum_{\tilde{\theta}} \sigma(p \mid \tilde{\theta}) \Prob[\tilde{\theta}] \\ &+\sum_{\theta} v_{\theta} \sigma(p \mid \theta) \Prob[ \theta] \sum_{\tilde{\theta}}( \tilde{\sigma}_{\eta}(p' \mid \tilde{\theta}) - \sigma(p \mid \tilde{\theta}))  \Prob[\tilde{\theta}]  \biggr\rvert \\ & \leq \frac{1}{\sigma(p) \cdot (\gamma - \tilde{\varepsilon})} \biggl(\overbrace{\biggl\lvert \sum_{\theta} v_{\theta} (\sigma(p \mid \theta) - \tilde{\sigma}_{\eta}(p' \mid \theta)) \sum_{\tilde{\theta}} \sigma(p \mid \tilde{\theta}) \Prob[\tilde{\theta}] \biggr \rvert}^{\leq \overline{v} \tilde{\varepsilon} \sigma(p) } \\ &+ \overbrace{\biggl \lvert \sum_{\theta} v_{\theta} \sigma(p \mid \theta) \Prob[ \theta] \sum_{\tilde{\theta}}( \tilde{\sigma}_{\eta}(p' \mid \tilde{\theta}) - \sigma(p \mid \tilde{\theta}))  \Prob[\tilde{\theta}]  \biggr\rvert}^{\leq \overline{v} \cdot \tilde{\varepsilon} \cdot \sigma(p)} \biggr) \\ & \leq \frac{2 \overline{v} \tilde{\varepsilon}}{\gamma - \tilde{\varepsilon}}. 
\end{align*}
The first inequality follows from adding and subtracting $\sum_{\theta} v_{\theta}\sigma(p \mid \theta) \Prob[\theta] \sum_{\tilde{\theta}} \sigma(p \mid \tilde{\theta}) \Prob[\tilde{\theta}]$ to the numerator inside the absolute value (as well as the lower bound on $\tilde{\sigma}_{\eta}(p)$), and the second inequality is from the triangle inequality, and the overbraced expression follows from  $v_{\theta} \leq \overline{v}$ and uniform convergence of $\tilde{\sigma}_{\eta}$ to $\sigma$.

So for any fixed $\gamma$, we can find some some $\eta$ such that whenever $\eta < \overline{\eta}$, we can ensure that on $K_{\gamma}$, $\abs{\E[v_{\theta} \mid \tilde{\sigma}_{\eta}, p] - \E[v_{\theta} \mid \sigma, p]}  < \varepsilon^{*}$, by choosing $\tilde{\varepsilon}$ sufficiently small so that $\frac{2 \tilde{\varepsilon}}{\gamma (\gamma - \tilde{\varepsilon})} < \varepsilon^{*}$. It follows that if the receiver's classifier converges to a rule that is $\varepsilon$-optimal under $\tilde{\sigma}_{\eta}$, it converges to a rule that is $\varepsilon+\varepsilon^{*}$ optimal under $\sigma$. The probability that this fails to occur is simply the probability that the price is outside of $K_{\gamma}$, which can be made arbitrarily small by taking $\gamma \rightarrow 0$, since we can approximate the support of $\sigma$ arbitrarily well.

\subsubsection{Step Three}
Note that, for an arbitrary continuous distribution $f$, if $p \sim f$ we have (for any compact $K$):

\begin{equation*} 
\Prob_{f}[L_{\gamma}] = \int_{K} \mathbf{1}[p: f(p) \leq \gamma] f(p) dp \leq  \int_{K} \mathbf{1}[p: f(p) \leq \gamma] \gamma dp \leq \mu(K) \cdot \gamma,
\end{equation*} 

\noindent where $\mu$ is Lebesgue measure. It follows that the probability that $p \in L_{\gamma}$, is small if $\gamma$ is small, and furthermore that this probability can be made small uniformly, using only $\gamma$. 

As shown by the claim above, by taking $\eta$ small, we can ensure that the difference in the conditional expected value is small with high probability. By taking $\gamma$ small, we ensure that the probability of a different outcome due to smoothing goes to 0, implying the result.

\section{Proofs for Examples} \label{sect: strictly dominated}
\begin{proof}[Proof of Lemma \ref{lm: strictly dominated}]

It suffices to show that 
if $p>v_L$ and $\Expect(v|p)-p\ge 0$, then the expected profit from $p$ is strictly less than $\pi_L v_L$.   We write the proof in \citeasnoun{Rubinstein93} for the later reference.
For any price $p$ satisfying
\[
\Prob(H |p) v_H + \Prob(L|p)v_L \ge p,
  \]
  the revenue cannot exceed
  \[
\Prob(H|p)  v_H+\Prob (L|p)  v_L 
  \]
  but the cost is
  \[
\Prob(H | p) (1-r) c_2 +\Prob(H|p) r c_1.
  \]
  Thus, the sender's expected payoff is at most
  \[
\Prob(L|p)  v_L + \Prob(H|p) ((1-r)(v_H-c_2) +r(v_H-c_1))
  \]
Because of the lemon's problem,
\[
(1-r)(v_H-c_2) +r(v_H-c_1)<0
\]
and
\[
\Prob(H|p)>0
\]
to satisfy
\[
\Prob(H |p) v_H + \Prob(L|p)v_L \ge p > v_L.
  \]
Integrating over $p$, we conclude that the ex ante profit is strictly less than
$\pi_L v_L$.

\end{proof}

\bibliographystyle{econometrica}
\bibliography{adaboost}

\end{document}